\documentclass[11pt]{article}

\usepackage[utf8]{inputenc} 
\usepackage[T1]{fontenc}
\usepackage{lmodern}
\usepackage{hyperref}       
\usepackage{url}            
\usepackage{booktabs}       
\usepackage{amsfonts}       
\usepackage[numbers]{natbib}
\usepackage{nicefrac}       
\usepackage{microtype}      
\usepackage{algorithm,algpseudocode}
\usepackage{amsfonts}
\usepackage{amsmath}
\usepackage{amsthm}
\usepackage{amssymb}
\usepackage[title]{appendix}
\usepackage{bm}
\usepackage{courier}
\usepackage[usenames,dvipsnames]{color}
\usepackage{enumitem}
\usepackage{graphicx}
\usepackage[margin=0.9in]{geometry}
\usepackage{url}
\usepackage{subfigure}
\usepackage{multirow}
\usepackage[dvipsnames]{xcolor}
\hypersetup{
	colorlinks,
	linkcolor={red!80!black},
	citecolor={blue!50!black},
	urlcolor={blue!80!black}
}

\input{./Definitions}

\graphicspath{ {./figures/} }
\DeclareGraphicsExtensions{.pdf, .jpeg, .png}
\allowdisplaybreaks

\title{
Speeding up Linear Programming using Randomized Linear Algebra
}

\author{
Agniva Chowdhury%
\footnote{Department of Statistics, Purdue University,
West Lafayette, IN, USA, \url{chowdhu5@purdue.edu}.}
\and
Palma London%
\footnote{Department of Computer Science, California Institute of Technology,
	Pasadena, CA, USA, \url{plondon@caltech.edu}.}
\and
Haim Avron%
\footnote{School of Mathematical Sciences, Tel Aviv University,
	Tel Aviv, Israel, \url{haimav@tauex.tau.ac.il}.}
\and
Petros Drineas%
\footnote{Department of Computer Science, Purdue University,
West Lafayette, IN, USA, \url{pdrineas@purdue.edu}.}
}

\date{}

\begin{document}

\maketitle


\begin{abstract}
	 Linear programming (LP) is an extremely useful tool and has been successfully applied to solve various problems in a wide range of areas, including operations research, engineering, economics, or even more abstract mathematical areas such as combinatorics. It is also used in many machine learning applications, such as $\ell_1$-regularized SVMs, basis pursuit, nonnegative matrix factorization, etc.  Interior Point Methods (IPMs) are one of the most popular methods to solve LPs both in theory and in practice. Their underlying complexity is dominated by the cost of solving a system of linear equations at each iteration.  In this paper, we consider \emph{infeasible} IPMs for the special case where the number of variables is much larger than the number of constraints. Using tools from Randomized Linear Algebra, we present a preconditioning technique that, when combined with the Conjugate Gradient iterative solver, provably guarantees that infeasible IPM algorithms (suitably modified to account for the error incurred by the approximate solver), converge to a feasible, approximately optimal solution, without increasing their iteration complexity. Our empirical evaluations verify our theoretical results on both real-world and synthetic data.
\end{abstract}


\section{Introduction}\label{sec:intro}

Linear programming (LP) is one of the most useful tools available to theoreticians and practitioners throughout science and engineering. It has been extensively used to solve various problems in a wide range of areas, including operations research, engineering, economics, or even in more abstract mathematical areas such as combinatorics. Also in machine learning and numerical optimization, LP appears in numerous settings, including $\ell_1$-regularized SVMs~\cite{zhu20041}, basis pursuit (BP)~\cite{yang2011alternating}, sparse inverse covariance matrix estimation (SICE)~\cite{yuan2010high}, the nonnegative matrix factorization (NMF)~\cite{recht2012factoring}, MAP inference~\cite{meshi2011alternating}, etc. Not surprisingly, designing and analyzing LP algorithms is a topic of paramount importance in computer science and applied mathematics.

One of the most successful paradigms for solving LPs is the family of Interior Point Methods (IPMs), pioneered by Karmarkar in the mid 1980s~\cite{karmarkar84}. Path-following IPMs and, in particular, long-step path following IPMs, are among the most practical approaches for solving linear programs. Consider the standard form of the primal LP problem:
\begin{flalign}
	\min\,\cbb^\ts\xb\,,\text{ subject to }\Ab\xb=\bb\,,\xb\ge \zero\,,\label{eq:primal}
\end{flalign}
where $\Ab\in\RR{\dimone}{\dimtwo}$, $\bb\in\R{\dimone}$, and $\cbb\in\R{\dimtwo}$ are the inputs, and $\xb\in\R{\dimtwo}$ is the vector of the primal variables. The associated dual problem is
\begin{flalign}
	\max\,\bb^\ts\yb\,,\text{ subject to }\Ab^\ts\yb+\sbb=\cbb\,,\sbb\ge\zero\,,\label{eq:dual}
\end{flalign}
where $\yb\in\R{\dimone}$ and $\sbb\in\R{\dimtwo}$ are the vectors of the dual and slack variables respectively.
Triplets $(\xb, \yb, \sbb)$ that uphold both \eqref{eq:primal} and \eqref{eq:dual} are called \emph{primal-dual solutions}.  Path-following IPMs typically
converge towards a primal-dual solution by operating as follows: given the current iterate $(\xb^{k},\yb^{k},\sbb^{k})$, they compute the Newton search direction $(\Delta\xb,\Delta\yb,\Delta\sbb)$ and update the current iterate by following a step towards the search direction. To compute the search direction, one standard approach~\cite{NW06} involves solving the
\emph{normal equations}\footnote{Another widely used approach is to solve the augmented system~\cite{NW06}. This approach is less relevant for this paper.}:
\begin{flalign}
	\Ab\Db^2\Ab^\ts\Delta\yb=~&\pb.\label{eq:normal}
\end{flalign}
Here, $\Db = \Xb^{1/2}\Sb^{-1/2}$ is a diagonal matrix, $\Xb,\Sb\in\RR{\dimtwo}{\dimtwo}$ are diagonal matrices whose $i$-th diagonal entries are equal to $\xb_i$ and $\sbb_i$, respectively, and $\pb \in \R{\dimone}$ is a vector whose exact definition is given in eqn.~(\ref{eqn:pdef})\footnote{The superscript $k$ in eqn.~(\ref{eqn:pdef}) simply indicates iteration count and is omitted here for notational simplicity.}. Given $\Delta\yb$, computing $\Delta \sbb$ and $\Delta \xb$ only involves matrix-vector products.

The core computational bottleneck in IPMs is the need to solve the linear system of eqn.~(\ref{eq:normal}) at each iteration. This leads to two key challenges: first, for high-dimensional matrices $\Ab$, solving the linear system is computationally prohibitive. Most implementations of IPMs use a \emph{direct solver}; see Chapter 6 of~\cite{NW06}. However, if $\Ab\Db^2\Ab^\ts$ is large and dense, direct solvers are computationally impractical. If $\Ab\Db^2\Ab^\ts$ is sparse, specialized direct solvers have been developed, but these do not apply to many LP problems, 
especially those arising in machine learning applications,
due to irregular sparsity patterns. 
%
%
Second, an alternative to direct solvers is the use of iterative solvers, but the situation is further complicated since $\Ab\Db^2\Ab^\ts$ is typically ill-conditioned. Indeed, as IPM algorithms approach the optimal primal-dual solution, the diagonal matrix $\Db$ becomes ill-conditioned, which also results in the matrix $\Ab\Db^2\Ab^\ts$ becoming ill-conditioned. Additionally, using approximate solutions for the linear system of eqn.~(\ref{eq:normal}) causes certain invariants, which are crucial for guaranteeing the convergence of IPMs, to be violated; see Section~\ref{sxn:contrib} for details.

In this paper, we address the aforementioned challenges, for the special case where $m \ll n$, 
i.e., the number of constraints is much smaller than the number of variables; see Section~\ref{sxn:extensions} for a generalization. This is a common setting in 
many applications 
of LP solvers. For example, in machine learning, 
$\ell_1$-SVMs and basis pursuit problems often exhibit such structure when the number of available features ($n$)  is larger than the number of objects ($m$). Indeed, this setting has been of interest in recent work on LPs \cite{Donoho05,Bienstock06,LondonAAAI2018}.  For simplicity of exposition, we also assume that the constraint matrix $\Ab$ has full rank, equal to $m$. First, we propose and analyze a preconditioned Conjugate Gradient (CG) iterative solver for the normal equations of eqn.~(\ref{eq:normal}), using matrix sketching constructions from the Randomized Linear Algebra (RLA) literature. We develop a preconditioner for $\Ab\Db^2\Ab^\ts$
%
%
using matrix sketching which allows us to prove strong convergence guarantees for the \textit{residual}
of CG solvers. 
%
Second, building upon the work of~\cite{Mon03}, we propose and analyze a provably accurate long-step \textit{infeasible} IPM algorithm. The proposed IPM solves the normal equations using iterative solvers. In this paper, for brevity and clarity, we primarily focus our description and analysis on the CG iterative solver. We note that a non-trivial concern is that the use of iterative solvers and matrix sketching tools implies that the normal equations at each iteration will be solved only approximately. In our proposed IPM, we develop a novel way to \textit{correct} for the error induced by the approximate solution in order to guarantee convergence. Importantly, this correction step is relatively computationally light, unlike a similar step proposed in~\cite{Mon03}. Third, we empirically show that our algorithm performs well in practice. We consider solving LPs that arise from $\ell_1$-regularized SVMs and test them on a variety of synthetic and real-world data sets. Several extensions of our work are discussed in Section~\ref{sxn:extensions}.

\subsection{Our contributions}\label{sxn:contrib}

Our point of departure in this work is the introduction of preconditioned, iterative solvers for solving eqn.~(\ref{eq:normal}). Preconditioning is used to address the ill-conditioning of the matrix $\Ab\Db^2\Ab^\ts$. Iterative solvers allow the computation of approximate solutions using only matrix-vector products while avoiding matrix inversion, Cholesky or LU factorizations, etc. A preconditioned formulation of eqn.~(\ref{eq:normal}) is:
\begin{flalign}
	\Qb^{-1}\Ab\Db^2\Ab^\ts\Delta\yb=\Qb^{-1}\pb,
	\label{eq:precond}
\end{flalign}
where $\Qb \in \mathbb{R}^{m \times m}$ is the preconditioning matrix; $\Qb$ should be easily invertible (see~\cite{axelsson1984finite,GVL12} for background). An alternative yet equivalent formulation of eqn.~(\ref{eq:precond}), which is more amenable to theoretical analysis, is
\begin{flalign}
	\Qb^{-\nicefrac{1}{2}}\Ab\Db^2\Ab^\ts\Qb^{-\nicefrac{1}{2}}\zb=~\Qb^{-\nicefrac{1}{2}}\pb,\label{eq:precond_alt}
\end{flalign}
where $\zb\in\R{\dimone}$ is a vector such that $\Delta\yb=\Qb^{-\nicefrac{1}{2}}\zb$. Note that the matrix in the left-hand side of the above equation is always symmetric, which is not necessarily the case for eqn.~\eqref{eq:precond}. We do emphasize that one can use eqn.~\eqref{eq:precond} in the actual implementation of the preconditioned solver; eqn.~(\ref{eq:precond_alt}) is much more useful in theoretical analyses. 

Recall that we focus on the special case where $\Ab \in \mathbb{R}^{m \times n}$ has $m \ll n$, i.e., it is a short-and-fat matrix. Our first contribution starts with the design and analysis of a preconditioner for the Conjugate Gradient solver that satisfies, with high probability,
\begin{flalign}\label{eq:pdcond1}
	\frac{2}{2+\zeta} \leq \sigma^2_{\min}(\Qb^{-\frac{1}{2}}\Ab\Db) \leq \sigma^2_{\max}(\Qb^{-\frac{1}{2}}\Ab\Db) \leq \frac{2}{2-\zeta},
\end{flalign}
for some error parameter $\zeta \in  [0,1]$. In the above, $\sigma_{\min}(\cdot)$ and $\sigma_{\max}(\cdot)$ correspond to the smallest and largest singular value of the matrix in parentheses.
The above condition says that the preconditioner effectively reduces the condition number of $\Ab\Db$ to a constant. We note that the particular form of the lower and upper bounds in eqn.~(\ref{eq:pdcond1}) was chosen to simplify our derivations. RLA matrix-sketching techniques allow us to construct preconditioners for all short-and-fat matrices that satisfy the above inequality \textit{and} can be inverted efficiently. Such constructions go back to the work of~\cite{Avron2010}; see Section~\ref{sxn:PCG} for details on the construction of $\Qb$ and its inverse. Importantly, given such a preconditioner, we then prove that the resulting CG iterative solver satisfies
\begin{flalign}
	\|\Qb^{-\nicefrac{1}{2}}\Ab\Db^2\Ab^\ts\Qb^{-\nicefrac{1}{2}}\tilde{\zb}^t-\Qb^{-\nicefrac{1}{2}}\pb\|_2\leq
	\zeta^t \|\Qb^{-\nicefrac{1}{2}}\pb\|_2. \label{eq:pdcond2}
\end{flalign}
Here $\tilde{\zb}^t$ is the approximate solution returned by the CG iterative solver after $t$ iterations. In words, the above inequality states that the \textit{residual}
achieved after $t$ iterations of the CG iterative solver drops exponentially fast. To the best of our knowledge, this result is not known in the CG literature: indeed, it is actually well-known that the residual error of CG may oscillate, even in cases where the energy norm of the solution error decreases monotonically. However, we prove that if the preconditioner is sufficiently good, i.e., it satisfies the constraint of eqn.~\eqref{eq:pdcond1}, then the residual error decreases as well.

Our second contribution is the analysis of a novel variant of a long-step \textit{infeasible} IPM algorithm proposed by~\cite{Mon03}. Recall that such algorithms can, in general, start with an initial point that is not necessarily feasible, but does need to satisfy some, more relaxed, constraints. Following the lines of~\cite{Zh94,Mon03}, let $\mathcal{S}$ be the set of feasible and optimal solutions
of the form $(\xb^*,\yb^*,\sbb^*)$ for the primal and dual problems of eqns.~\eqref{eq:primal} and~\eqref{eq:dual} and assume that $\mathcal{S}$ is not empty. Then, long-step infeasible IPMs can start with any initial point $(\xb^{0},\yb^{0},\sbb^{0})$ that satisfies $(\xb^{0},\sbb^{0}) > 0$ \textit{and} $(\xb^{0},\sbb^{0}) \geq (\xb^{*},\sbb^{*})$, for some feasible and optimal solution
$(\xb^{*},\sbb^{*})\in \mathcal{S}$. In words, the starting primal and slack variables must be strictly positive \textit{and} larger (element-wise) when compared to some feasible, optimal primal-dual solution. See Chapter 6 of \cite{wright1997primal}
for a discussion regarding why such choices of starting points 
are
also relevant to computational practice.

The flexibility of infeasible IPMs comes at a cost: long-step \textit{feasible} IPMs converge in $\Ocal(n\log\nicefrac{1}{\epsilon})$  iterations, while long-step \textit{infeasible} IPMs need $\Ocal(n^2 \log\nicefrac{1}{\epsilon})$ iterations to converge~\cite{Zh94,Mon03}. Here $\epsilon$ is the accuracy of the approximate LP solution returned by the IPM; see Algorithm~\ref{algo:iipm} for the exact definition. Let
\begin{flalign}
	\Ab\xb^0-\bb&= \rb_p^0, \label{eq:primalres}\\
	\Ab^\ts\yb^0+\sbb^0-\cbb &= \rb_d^0,\label{eq:dualres}
\end{flalign}
where $\rb_p^0 \in \mathbb{R}^n$ and $\rb_d^0 \in \mathbb{R}^m$ are the \textit{primal} and \textit{dual} residuals, respectively, and characterize how far the initial point is from being feasible. As long-step infeasible IPM algorithms iterate and update the primal and dual solutions, the residuals are updated as well. Let $\rb^k = (\rb_p^k,\rb_d^k) \in \mathbb{R}^{n+m}$ be the primal and dual residual at the $k$-th iteration: it is well-known that the convergence analysis of infeasible long-step IPMs critically depends on $\rb^k$ lying on the line segment between 0 and $\rb^0$. Unfortunately, using approximate solvers (such as the CG solver proposed above) for the normal equations violates this invariant.Aa simple solution to fix this problem by adding a perturbation vector $\vb$ to the current primal-dual solution that guarantees that the invariant is satisfied is proposed in~\cite{Mon03}. 
%
%
Again, we use RLA matrix sketching principles to propose an efficient construction for $\vb$ that provably satisfies the invariant. Next, we combine the above two primitives to prove that Algorithm~\ref{algo:iipm} in Section~\ref{sxn:IIPM} satisfies the following theorem.
\begin{theorem}\label{thm:1}
	Let $0 \leq \epsilon \leq 1$ be an accuracy parameter. Consider the long-step infeasible IPM Algorithm~\ref{algo:iipm} (Section~\ref{sxn:IIPM}) that solves eqn.~(\ref{eq:precond_alt}) using the CG solver of Algorithm~\ref{algo:PCG} (Section~\ref{sxn:PCG}). Assume that the CG iterative solver runs with accuracy parameter $\zeta = \nicefrac{1}{2}$ and iteration count
	$t = \Ocal (\log n)$. 
	%
	Then, with probability at least 0.9, the long-step infeasible IPM converges after $\Ocal(n^2 \log \nicefrac{1}{\epsilon})$ iterations.
\end{theorem}
We note that the 0.9 success probability above is for simplicity of exposition and can be easily amplified using standard techniques. Also, at each iteration of our infeasible long-step IPM algorithm, the running time is $\Ocal((\nnz{A}+m^3)\log n)$. See Section~\ref{sxn:IIPM} for a detailed discussion of the overall running time.

Our empirical evaluation demonstrates that our algorithm requires an order of magnitude much fewer inner CG iterations than a standard IPM using CG, while producing a comparably accurate solution (see Section~\ref{sec:exp}). In practice, our empirical evaluation also indicates that using a CG solver with our sketching-based preconditioner does not increase the number of (outer) iterations of the infeasible IPM, compared to unpreconditioned CG or  a direct linear solver. In particular, there are instances where our solver performs much better than unpreconditioned CG in terms of (outer) iteration count.

\subsection{Comparison with Related Work}\label{sxn:comparison}

There is a large body of literature on solving LPs using IPMs. We only review literature that is immediately relevant to our work. Recall that we solve the normal equations inexactly at each iteration, and develop a way to \emph{correct} for the error incurred. We also focus on IPMs that can use an sufficiently positive, infeasible initial point (see Section~\ref{sxn:contrib}). We discuss below two papers that present related ideas.

The use of an approximate iterative solver for eqn.~(\ref{eq:normal}), followed by a correction step to ``fix'' the approximate solution was proposed in~\cite{Mon03} (see our discussion in Section~\ref{sxn:contrib}). We propose efficient, RLA-based approaches to precondition and solve eqn.~(\ref{eq:normal}), as well as a novel approach to correct for the approximation error in order to guarantee the convergence of the IPM algorithm. Specifically,~\cite{Mon03} propose to solve eqn.~\eqref{eq:normal} using the so-called \emph{maximum weight basis} preconditioner \cite{RV93}.
However, computing such a preconditioner needs access to a maximal linearly independent set of columns of $\Ab\Db$ in each iteration, which is costly, taking $\Ocal(\dimone^2\dimtwo)$ time in the worst-case.
More importantly, while~\cite{Mon04} was able to provide a bound on the condition number of the preconditioned matrix, that depends only on properties of $\Ab$, and is independent of $\Db$, this bound might, in general, be very large. In contrast, our bound is a constant and it does not depend on properties of $\Ab$ or its dimension. In addition, \cite{Mon03} assumed a bound on the two-norm of the residual of the preconditioned system, but it is unclear how their preconditioner guarantees such a bound. Similar concerns exist for the construction of the correction vector $\vb$ proposed by~\cite{Mon03}, which our work alleviates.

The line of research in the Theoretical Computer Science literature that is closest to our work is~\cite{daitch2008faster}, who presented an IPM that uses an approximate solver in each iteration. However, their accuracy guarantee is in terms of the final objective value which is different from ours.
More importantly,~\cite{daitch2008faster} focuses on \textit{short-step}, feasible IPMs, whereas ours is \emph{long-step} and does not require a feasible starting point. Finally, the approximate solver proposed by \cite{daitch2008faster} works only for the special case of input matrices that correspond to graph Laplacians, following the lines of~\cite{spielman2004nearly,spielman2014nearly}.

We also note that in the Theoretical Computer Science literature,~\cite{LS13,lee2013path2, LS14,LS15,lee2019solving,CLS19} proposed and analyzed theoretically ground-breaking algorithms for LPs based on novel tools such as the so-called \emph{inverse maintenance} for accelerating the linear system solvers in IPMs.  However, all these endeavors are primarily focused on the theoretically fast but practically inefficient short-step feasible IPMs. In contrast, our work is focused on infeasible \textit{long-step} IPMs, known to work efficiently in practice. Very recently,~\cite{brand2020solving} proposed another fast, short-step, feasible IPM for solving tall and dense LPs. The output of their algorithm does not satisfy the linear constraints exactly (similar to~\cite{daitch2008faster}) and  the final convergence guarantee is somewhat different from our work.

Another relevant line of research is the work of~\cite{CMTH16}, which proposed solving eqn.~\eqref{eq:normal}  using preconditioned Krylov subspace methods, including variants of \emph{generalized minimum residual} (GMRES) or CG methods. Indeed, \cite{CMTH16} conducted extensive numerical experiments on LP problems taken from standard benchmark libraries, but did not provide any theoretical guarantees.

From a matrix-sketching perspective, our work was partially motivated by~\cite{CYD18}, which presented an iterative, sketching-based algorithm to solve under-constrained ridge regression problems, but did not address how to make use of such approaches in an IPM-based framework, as we do here. Recent papers proposed the so-called \emph{Newton sketch}~\cite{PW17,XYRCM16} to construct an approximate Hessian matrix  for more general convex objective functions of which LP is a special case. Nevertheless, these randomized second-order methods are significantly faster than the conventional approach only when the data matrix is over-constrained, \ie~$\dimone\gg\dimtwo$. It is unclear whether the approach of~\cite{PW17,XYRCM16} is faster than IPMs when the optimization problem to be solved is linear. A probabilistic algorithm to solve LP approximately in a random projection-based reduced feature-space was proposed in ~\cite{VPL18}. A possible drawback of this paper is that the approximate solution is infeasible with respect to the original region.

Finally, we refer the interested reader to the surveys~\cite{Woodruff14, DM2018,Halko2011,Mahoney11,Drineas2016} for more background on Randomized Linear Algebra.


\section{Notation and Background}\label{sxn:background}

$\mathbf{A}, \mathbf{B}, \ldots$ denote
matrices and $\mathbf{a}, \mathbf{b}, \ldots$ denote vectors. For vector $\ab$, $\|\mathbf{a}\|_{2}$ denotes its Euclidean norm; for a matrix $\mathbf{A},\|\mathbf{A}\|_{2}$ denotes
its spectral norm and $\|\Ab\|_F$ denotes its Frobenius norm. We use $\zero$ to denote a null vector or null matrix, dependent upon context, and $\one$ to denote the all-ones vector.
For any matrix $\Xb\in\RR{\dimone}{\dimtwo}$ with $\dimone\leq\dimtwo$ of rank $\dimone$ a thin Singular Value Decomposition (SVD) is a product $\Ub\Sigmab\Vb^\ts$ , with $\mathbf{U} \in\mathbb{R}^{\dimone\times\dimone}$ (the matrix of the left singular vectors), $\mathbf{V} \in$ $\mathbb{R}^{\dimtwo \times \dimone}($ the matrix of the top-$\dimone$ right singular vectors), and $\Sigmab \in$
$\mathbb{R}^{\dimone \times \dimone}$ a diagonal matrix whose entries are equal to the singular values of $\Xb$. We use $\sigma_{i}(\cdot)$ to denote the $i$-th singular value of the matrix in parentheses. 

For any two symmetric positive semidefinite (resp. positive definite) matrices $\Ab_1$ and $\Ab_2$ of appropriate dimensions, $\Ab_1\preccurlyeq\Ab_2$ ($\Ab_1\prec\Ab_2$) denotes that $\Ab_2-\Ab_1$ is positive semidefinite (resp. positive definite).


We now briefly discuss a result on matrix sketching~\cite{Cohen2016,cohen2016nearly} that is particularly useful in our theoretical analyses. In our parlance,~\cite{Cohen2016} proved that,
for any matrix $\Zb\in\RR{m}{n}$, there exists a sketching matrix $\Wb\in\RR{n}{w}$ such that
\begin{flalign}\label{eqn:pdprec}
	\nbr{\Zb \Wb \Wb^\ts \Zb^\ts - \Zb \Zb^\ts}_2\le \frac{\zeta}{4}\Big(\nbr{\Zb}_2^2+\frac{\|\Zb\|_F^2}{r}\Big)
\end{flalign}
holds with probability at least $1-\delta$ for any $r\ge 1$. Here $\zeta \in [0,1]$ is a (constant) accuracy parameter. Ignoring constant terms, $w=\Ocal(r\log(\nicefrac{r}{\delta}))$; $\Wb$ has $\Ocal(\log(r/\delta))$ non-zero entries per row; and the product $\Zb\Wb$ can be computed in time
$\Ocal(\log(r/\delta)\cdot\nnz{\Zb})$.


\section{Conjugate Gradient Solver}\label{sxn:PCG}
%
In this section, we discuss the computation of the preconditioner $\Qb$ (and its inverse), followed by a discussion on how such a preconditioner can be used to satisfy eqns.~\eqref{eq:pdcond1} and~\eqref{eq:pdcond2}.
\begin{algorithm}[H]
	\caption{Solving eqn.~\eqref{eq:precond_alt} via CG}\label{algo:PCG}
	\begin{algorithmic}
		\State \textbf{Input:}
		$\Ab\Db\in\RR{\dimone}{\dimtwo}$, $\pb\in\R{\dimone}$, sketching matrix $\Wb \in \mathbb{R}^{n \times w}$,  iteration count $t$;
		\vspace{1mm}
		\State 1. Compute $\Ab\Db\Wb$ and its SVD: let $\Ub_{\Qb} \in \mathbb{R}^{m \times m}$ be the matrix of its left singular vectors and let $\Sigmab_{\Qb}^{\nicefrac{1}{2}} \in \mathbb{R}^{m \times m}$ be the matrix of its singular values;
		\vspace{1mm}
		\State 2. Compute $\Qb^{-\nicefrac{1}{2}} = \Ub_{\Qb} \Sigmab_{\Qb}^{-\nicefrac{1}{2}}\Ub_{\Qb}^\ts$;
		\vspace{1mm}
		\State 3. Initialize $\tilde{\zb}^{0} \gets \zero_\dimone $ and run standard CG on the preconditioned system of eqn.~\eqref{eq:precond_alt} for $t$ iterations;
		\vspace{1mm}
		\State \textbf{Output:} return $\tilde{\zb}^t$;
	\end{algorithmic}
\end{algorithm}
\noindent Algorithm~\ref{algo:PCG} takes as input the sketching matrix $\Wb \in \mathbb{R}^{n \times w}$, which we construct as discussed in Section~\ref{sxn:background}. Our preconditioner $\Qb$ is equal to
\begin{flalign}\label{eqn:pdprecond}
	\Qb=\Ab\Db\Wb\Wb^\ts\Db\Ab^\ts.
\end{flalign}
Notice that we only need to compute $\Qb^{\nicefrac{-1}{2}}$ in order to use it to solve eqn.~(\ref{eq:precond_alt}). Towards that end, we first compute the sketched matrix $\Ab\Db\Wb \in \mathbb{R}^{\dimone \times w}$. Then, we compute the SVD of the matrix $\Ab\Db\Wb$: let $\Ub_{\Qb}$ be the matrix of its left singular vectors and let $\Sigmab_{\Qb}^{\nicefrac{1}{2}}$ be the matrix of its singular values. Notice that the left (and right) singular vectors of $\Qb^{\nicefrac{-1}{2}}$ are equal to $\Ub_{\Qb}$ and its singular values are equal to $\Sigmab_{\Qb}^{-\nicefrac{1}{2}}$. Therefore, $\Qb^{\nicefrac{-1}{2}} = \Ub_{\Qb} \Sigmab_{\Qb}^{-\nicefrac{1}{2}}\Ub_{\Qb}^\ts$.

Let $\Ab\Db = \Ub\Sigmab\Vb^\ts$ be the thin SVD representation of $\Ab\Db$. We apply the results of~\cite{Cohen2016} (see Section~\ref{sxn:background}) to the matrix $\Zb = \Vb^\ts \in\mathbb{R}^{m \times n}$ with $r=m$ to get that, with probability at least $1-\delta$,
\begin{flalign}
	\nbr{\Vb^\ts \Wb \Wb^\ts \Vb - \Ib_{\dimone}}_2\le \frac{\zeta}{4}\Big(\nbr{\Vb}_2^2+\frac{\|\Vb\|_F^2}{m}\Big) \le \frac{\zeta}{2}.\label{eq:cnd1}
\end{flalign}
In the above we used $\nbr{\Vb}_2=1$ and $\nbr{\Vb}_F^2=m$. The running time needed to compute the sketch $\Ab\Db\Wb$ is equal to (ignoring constant factors)
%
$\Ocal(\nnz{\Ab}\cdot \log(m/\delta))$.
%
Note that $\nnz{\Ab\Db}=\nnz{\Ab}$. The cost of computing the SVD of $\Ab\Db\Wb$ (and therefore $\Qb^{\nicefrac{-1}{2}}$) is $\Ocal(m^3\log(m/\delta))$. Overall, computing $\Qb^{-\nicefrac{1}{2}}$ can be done in time
\begin{flalign}\label{eqn:svdQ}
	\Ocal(\nnz{\Ab}\cdot \log(m/\delta)+m^3\log(m/\delta)).
\end{flalign}
Given these results, we now discuss how to satisfy eqns.~\eqref{eq:pdcond1} and \eqref{eq:pdcond2} using the sketching matrix $\Wb$. We start with the following bound, which is relatively straight-forward given prior RLA work.
%
\begin{lemma}\label{lem:cond3}
	If the sketching matrix $\Wb$ satisfies eqn.~\eqref{eq:cnd1}, then, for all $i=1\ldots m$,
	\begin{flalign*}
		(1+\zeta/2)^{-1}\le\sigma_i^2(\Qb^{-\nicefrac{1}{2}}\Ab\Db)\le (1-\zeta/2)^{-1}.
	\end{flalign*}
\end{lemma}
\begin{proof}
	Consider the condition of eqn.~\eqref{eq:cnd1}:
	\begin{flalign}
	&~\|\Vb^\ts\Wb\Wb^\ts\Vb-\Ib_\dimone\|_2\le\frac{\zeta}{2}~\Leftrightarrow~ -\frac{\zeta}{2}\,\Ib_{\dimone}\preccurlyeq\Vb^\ts\Wb\Wb^\ts\Vb-\Ib_\dimone\preccurlyeq\frac{\zeta}{2}\,\Ib_{\dimone}\label{eq:full}\\
	\Leftrightarrow~&-\frac{\zeta}{2}\,\Ab\Db^2\Ab^\ts\preccurlyeq\Ab\Db\Wb\Wb^\ts\Db\Ab^\ts-\Ab\Db^2\Ab^\ts\preccurlyeq\frac{\zeta}{2}\,\Ab\Db^2\Ab^\ts\label{eq:rich_3}\\
	\Leftrightarrow~&\left(1-\frac{\zeta}{2}\right)\,\Ab\Db^2\Ab^\ts\preccurlyeq\underbrace{\Ab\Db\Wb\Wb^\ts\Db\Ab^\ts}_{\Qb}\preccurlyeq\left(1+\frac{\zeta}{2}\right)\,\Ab\Db^2\Ab^\ts\,.\label{eq:rich_4}
	\end{flalign}
	We obtain eqn.~\eqref{eq:rich_3} by pre- and post-multiplying the previous inequality by $\Ub\Sigmab$ and $\Sigmab\Ub^\ts$ respectively and using the facts that $\Ab\Db=\Ub\Sigmab\Vb^\ts$ and $\Ab\Db^2\Ab^\ts=\Ub\Sigmab^2\Ub^\ts$. Also, from eqn.~\eqref{eq:full}, note that all the eigenvalues of $\Vb^\ts\Wb\Wb^\ts\Vb$ lie between $(1-\frac{\zeta}{2})$ and $(1+\frac{\zeta}{2})$ and thus $\rank(\Vb^\ts\Wb)=m$. Therefore, $\rank(\Ab\Db\Wb)=\rank(\Ub\Sigmab\Vb^\ts\Wb)=m$, as $\Ub\Sigmab$ is non-singular and we know that the rank of a matrix remains unaltered by pre- or post-multiplying it by a non-singular matrix. So, we have $\rank(\Qb)=m$; in words $\Qb$ has full rank. Therefore, all the diagonal entries of $\Sigmab_\Qb$ are positive and $\Qb^{-\nicefrac{1}{2}}\Qb\Qb^{-\nicefrac{1}{2}}=\Ib_m$\,.

Using the above arguments, pre- and post- multiplying eqn.~\eqref{eq:rich_4} by $\Qb^{-1/2}$, we get
	\begin{flalign}
	&~\left(1-\frac{\zeta}{2}\right)\,\Qb^{-1/2}\Ab\Db^2\Ab^\ts\Qb^{-1/2}\preccurlyeq\Ib_{m}\preccurlyeq\left(1+\frac{\zeta}{2}\right)\,\Qb^{-1/2}\Ab\Db^2\Ab^\ts\Qb^{-1/2}\nonumber\\
	\Rightarrow&~\left(1+\frac{\zeta}{2}\right)^{-1}\Ib_{m}\preccurlyeq\Qb^{-1/2}\Ab\Db^2\Ab^\ts\Qb^{-1/2}\preccurlyeq\left(1-\frac{\zeta}{2}\right)^{-1}\Ib_{\dimone}\,.\label{eq:normbound}
	\end{flalign}
Eqn.~\eqref{eq:normbound} implies that all the eigenvalues of $\Qb^{-1/2}\Ab\Db^2\Ab^\ts\Qb^{-1/2}$ are bounded between $\left(1+\frac{\zeta}{2}\right)^{-1}$ and $\left(1-\frac{\zeta}{2}\right)^{-1}$, which concludes the proof of the lemma.
\end{proof}
\noindent The above lemma directly implies eqn.~\eqref{eq:pdcond1}. We now proceed to show that the above construction for $\Qb^{\nicefrac{-1}{2}}$, when combined with the conjugate gradient solver to solve eqn.~\eqref{eq:precond_alt}, indeed satisfies eqn.~\eqref{eq:pdcond2}\footnote{See Chapter 9 of~\cite{Luenberger15} for a detailed overview of CG.}. We do note that in prior work most of the convergence guarantees for CG focus on the error of the approximate solution. However, in our work, we are interested in the convergence of the \textit{residuals} and it is known that even if the energy norm of the error of the approximate solution decreases monotonically, the norms of the CG residuals may oscillate. Interestingly, we can combine a result on the residuals of CG from~\cite{bouyouli2009new} with Lemma~\ref{lem:cond3} to prove that in our setting the norms of the CG residuals also decrease monotonically.

Let $\tilde{\fb}^{(j)}$ be the residual at the $j$-th iteration of the CG algorithm:
$$\tilde{\fb}^{(j)}=\Qb^{-\nicefrac{1}{2}}\Ab\Db^2\Ab^\ts\Qb^{-\nicefrac{1}{2}}\tilde{\zb}^j-\Qb^{-\nicefrac{1}{2}}\pb.$$
Recall from Algorithm~\ref{algo:PCG} that~$\tilde{\zb}^0=\zero$ and thus $\tilde{\fb}^{(0)}=-\Qb^{-\nicefrac{1}{2}}\pb$.
In our parlance, Theorem 8 of~\cite{bouyouli2009new} proved the following bound.
\begin{lemma}[Theorem 8 of \cite{bouyouli2009new}]\label{lem:prev1}%
	Let $\tilde{\fb}^{(j-1)}$ and $\tilde{\fb}^{(j)}$ be the residuals obtained by the CG solver at steps $j-1$ and $j$. Then,%
	\begin{flalign*}
	\|\tilde{\fb}^{(j)}\|_2\le~\frac{\kappa^2(\Qb^{-\nicefrac{1}{2}}\Ab\Db)-1}{2}\|\tilde{\fb}^{(j-1)}\|_2\,,
	\end{flalign*}
	where $\kappa(\Qb^{-\nicefrac{1}{2}}\Ab\Db)$ is the condition number of $\Qb^{-\nicefrac{1}{2}}\Ab\Db$.
\end{lemma}
\noindent From Lemma~\ref{lem:cond3}, we get
\begin{flalign}\label{eq:condbd}
\kappa^2(\Qb^{-\nicefrac{1}{2}}\Ab\Db)=\frac{\sigma_{\max}^2(\Qb^{-\nicefrac{1}{2}}\Ab\Db)}{\sigma_{\min}^2(\Qb^{-\nicefrac{1}{2}}\Ab\Db)}\le\frac{1+\zeta/2}{1-\zeta/2}.
\end{flalign}
Combining eqn.~\eqref{eq:condbd} with Lemma~\ref{lem:prev1},
\begin{flalign}
\|\tilde{\fb}^{(j)}\|_2\le~\frac{\frac{1+\zeta/2}{1-\zeta/2}-1}{2}\|\tilde{\fb}^{(j-1)}\|_2
=~\frac{\zeta}{2-\zeta}\|\tilde{\fb}^{(j-1)}\|_2
\le~\zeta \|\tilde{\fb}^{(j-1)}\|_2\label{eq:rec}\,,
\end{flalign}
where the last inequality follows from $\zeta\le1$. Applying eqn.~\eqref{eq:rec} recursively, we get
\begin{flalign}
\|\tilde{\fb}^{(t)}\|_2\le\zeta\|\tilde{\fb}^{(t-1)}\|_2\le\dots\le\zeta^t\|\tilde{\fb}^{(0)}\|_2=\zeta^t\|\Qb^{-\nicefrac{1}{2}}\pb\|_2\nonumber\,,
\end{flalign}
which proves the condition of eqn.~\eqref{eq:pdcond2}.

We remark that one can consider using MINRES~\cite{paige1975solution} instead of CG. Our results hinges on bounding the two-norm of the residual. MINRES finds, at each iteration, the optimal vector with respect the two-norm of the residual inside the same Krylov subspace of CG for the corresponding iteration. Thus, the bound we prove for CG applies to MINRES as well.


\section{The Infeasible IPM algorithm}\label{sxn:IIPM}

In order to avoid spurious solutions, primal-dual path-following IPMs bias the search direction towards the \emph{central path} and restrict the iterates to a neighborhood of the central path. This search is controlled by the \emph{centering parameter} $\sigma\in[0,1]$.
%
At each iteration, given the current solution $(\xb^{k},\yb^{k},\sbb^{k})$, a standard infeasible IPM obtains the search direction $(\Delta\xb^k,\Delta\yb^k,\Delta\sbb^k)$ by solving the following system of linear equations:
\begin{subequations}\label{eq:system}
	\begin{flalign}
		\Ab\Db^2\Ab^\ts\Delta\yb^k=~&\pb^k\,,\label{eq:normal1}\\
		\Delta\sbb^k=~&-\rb^k_d-\Ab^\ts\Delta\yb^k\,,\label{eq:dels}\\
		\Delta\xb^k=~&-\xb^k+\sigma\mu_k\Sb^{-1}\one_\dimtwo-\Db^2\Delta\sbb^k.\label{eq:delx}
	\end{flalign}
\end{subequations}
Here $\Db$ and $\Sb$ are computed given the current iterate $(\xb^{k}$ and $\sbb^{k})$; we skip the indices on $\Db$ and $\Sb$ for notational simplicity. 
After solving the above system, the infeasible IPM Algorithm~\ref{algo:iipm} proceeds by computing a step-size $\alphabar$ to return:
\begin{flalign}\label{eqn:update}
	(\xb^{k+1},\yb^{k+1},\sbb^{k+1}) = (\xb^{k},\yb^{k},\sbb^{k}) + \alphabar (\Delta \xb^k,\Delta \yb^k,\Delta \sbb^k).
\end{flalign}
Recall that $\rb^k=(\rb^k_p,\rb^k_d)$ is a vector with  $\rb^k_p=\Ab\xb^k-\bb$ and $\rb^k_d=\Ab^\ts\yb^k+\sbb^k-\cbb$ (the primal and dual residuals). We also use the \textit{duality measure} $\mu_k=\nicefrac{{\xb^k}^\ts\sbb^k}{n}$ and the vector
\begin{flalign}\label{eqn:pdef}
	\pb^k&=-\rb_p^k-\sigma\mu_k\Ab\Sb^{-1}\one_\dimtwo+\Ab\xb^k-\Ab\Db^2\rb^k_d.
\end{flalign}
Given $\Delta\yb^k$ from eqn.~(\ref{eq:normal1}), $\Delta\sbb^k$ and $\Delta\xb^k$ are easy to compute from eqns.~\eqref{eq:dels} and \eqref{eq:delx}, as they only involve matrix-vector products.  However, since we  use Algorithm~\ref{algo:PCG} to solve eqn.~\eqref{eq:normal1} approximately using the sketching-based preconditioned CG solver, the primal and dual residuals \textit{do not} lie on the line segment between $\zero$ and $\rb^0$. This invalidates known proofs of convergence for infeasible IPMs.

For notational simplicity, we now drop the dependency of vectors and scalars on the iteration counter $k$. Let $\hat{\Delta \yb}=\Qb^{\nicefrac{-1}{2}}\tilde{\zb}^t$ be the approximate solution to eqn.~(\ref{eq:normal1}). In order to account for the loss of accuracy due to the approximate solver, we compute $\hat{\Delta\xb}$ as follows:
\begin{flalign}
	\hat{\Delta\xb}=~-\xb+\sigma\mu\Sb^{-1}\one_\dimtwo-\Db^2\hat{\Delta\sbb}-\Sb^{-1}\vb\label{eq:delxhat}.
\end{flalign}
Here $\vb\in\R{n}$ is a perturbation vector that needs to exactly satisfy the following invariant at each iteration of the infeasible IPM:
\begin{flalign}
	\Ab\Sb^{-1}\vb=\Ab\Db^2\Ab^\ts\hat{\Delta\yb}-\pb\,\label{eq:addl}.
\end{flalign}
We note that the computation of $\hat{ \Delta \sbb}$ is still done using, essentially, eqn.~\eqref{eq:dels}, namely
\begin{flalign}
\Delta\hat{\sbb}^k=~&-\rb^k_d-\Ab^\ts\hat{\Delta\yb}^k.\label{eq:delshat}
\end{flalign}
In \cite{Mon03} it is argued that if $\vb$ satisfies eqn.~\eqref{eq:addl}, the primal and dual residuals lie in the correct line segment.

\vspace{0.02in}\noindent\textbf{Construction of $\vb$.} There are many choices for $\vb$ satisfying eqn.~\eqref{eq:addl}. 
To prove convergence, it is desirable for $\vb$ to have a small norm, hence a general choice is
$$\vb=(\Ab\Sb^{-1})^{\dagger}(\Ab\Db^2\Ab^\ts\hat{\Delta\yb}-\pb),$$
which involves the computation of the pseudoinverse $(\Ab\Sb^{-1})^{\dagger}$, which is expensive, taking time $\Ocal(\dimone^2\dimtwo)$. Instead, we propose
to construct $\vb$ using the sketching matrix $\Wb$ of Section~\ref{sxn:background}. More precisely, we construct the perturbation vector
\begin{flalign}\label{eq:compv}
	\vb=(\Xb\Sb)^{\nicefrac{1}{2}}\Wb(\Ab\Db\Wb)^{\dagger}(\Ab\Db^2\Ab^\ts\hat{\Delta\yb}-\pb).
\end{flalign}
The following lemma proves that the proposed $\vb$ satisfies eqn.~(\ref{eq:addl}).
%
\begin{lemma}\label{lem:fullrankR}
	Let $\Wb\in\RR{\dimtwo}{w}$ be the sketching matrix of Section~\ref{sxn:background} and $\vb$ be the perturbation vector of eqn.~(\ref{eq:compv}). Then, with probability at least $1-\delta$, $\rank(\Ab\Db\Wb)=\dimone$ and
	$\vb$ satisfies eqn.~\eqref{eq:addl}.
\end{lemma}
\begin{proof}
	Let $\Ab\Db=\Ub\Sigmab\Vb^\ts$ be the thin SVD representation of $\Ab\Db$.
	We use the exact same $\Wb$ as discussed in Section~\ref{sxn:PCG}.  Therefore, eqn.~\eqref{eq:cnd1} holds with probability $1-\delta$ and it directly follows from the proof of Lemma~\ref{lem:cond3} that $\rank(\Ab\Db\Wb)=\dimone$.
Recall that $\Ab\Db\Wb$ has full \emph{row-rank} and thus $\Ab\Db\Wb\,(\Ab\Db\Wb)^\dagger=\Ib_\dimone$. Therefore, taking $\vb=(\Xb\Sb)^{\nicefrac{1}{2}}\Wb(\Ab\Db\Wb)^{\dagger}(\Ab\Db^2\Ab^\ts\hat{\Delta\yb}-\pb)$, we get
	\begin{flalign*}
	\Ab\Sb^{-1}\,\vb=&~\Ab\Sb^{-1}(\Xb\Sb)^{\nicefrac{1}{2}}\Wb(\Ab\Db\Wb)^{\dagger}(\Ab\Db^2\Ab^\ts\hat{\Delta\yb}-\pb)\nonumber\\
	=&~\Ab\Db\Wb(\Ab\Db\Wb)^{\dagger}(\Ab\Db^2\Ab^\ts\hat{\Delta\yb}-\pb)\nonumber\\
	=&~\Ab\Db^2\Ab^\ts\hat{\Delta\yb}-\pb\,,
	\end{flalign*}
	where the second equality follows from $\Db = \Xb^{1/2}\Sb^{-1/2}$.
\end{proof}
 \noindent We emphasize here that we use the same exact sketching matrix $\Wb \in \mathbb{R}^{n \times w}$ to form the preconditioner used in the CG algorithm of Section~\ref{sxn:PCG} \textit{as well as} the vector $\vb$ in eqn.(\ref{eq:compv}). This allows us to sketch $\Ab \Db$  only once, thus saving time in practice. Next, we present a bound for the two-norm of the perturbation vector $\vb$ of eqn.~(\ref{eq:compv}).

\begin{lemma}\label{lem:v}
	With probability at least $1-\delta$, our perturbation vector $\vb$ in Lemma~\ref{lem:fullrankR} satisfies
	\begin{flalign}
		\|\vb\|_2\le\sqrt{3n\mu}\,\|\tilde{\fb}^{(t)}\|_2,
	\end{flalign}
	with $\tilde{\fb}^{(t)}=\Qb^{-\nicefrac{1}{2}}\Ab\Db^2\Ab^\ts\Qb^{-\nicefrac{1}{2}}\tilde{\zb}^t-\Qb^{-\nicefrac{1}{2}}\pb$.
\end{lemma}
\begin{proof}
Recall that $\Qb=\Ab\Db\Wb(\Ab\Db\Wb)^\ts=\Ub_\Qb\Sigmab_\Qb\Ub_\Qb^\ts$. Also, $\Ub_\Qb$ and $\Sigmab_\Qb^{\nicefrac{1}{2}}$ are (respectively) the matrices of the left singular vectors and the singular values of $\Ab\Db\Wb$. Now, let $\widehat{\Vb}$ be the right singular vector of $\Ab\Db\Wb$. Therefore, $\Ab\Db\Wb=\Ub_\Qb\Sigmab_\Qb^{\nicefrac{1}{2}}\widehat{\Vb}^\ts$ is the thin SVD representation of $\Ab\Db\Wb$. Also, from Lemma~\ref{lem:cond3}, we know that $\Qb$ has full rank. Therefore, $\Qb^{\nicefrac{1}{2}}\Qb^{\nicefrac{-1}{2}}=\Ib_m$. Next, we bound $\|\vb\|_2$:
\begin{flalign}
	\|\vb\|_2=&~\|(\Xb\Sb)^{\nicefrac{1}{2}}\Wb(\Ab\Db\Wb)^{\dagger}(\Ab\Db^2\Ab^\ts\hat{\Delta\yb}-\pb)\|_2\nonumber\\
	=&~\|(\Xb\Sb)^{\nicefrac{1}{2}}\Wb(\Ab\Db\Wb)^{\dagger}\Qb^{\nicefrac{1}{2}}\Qb^{\nicefrac{-1}{2}}(\Ab\Db^2\Ab^\ts\hat{\Delta\yb}-\pb)\|_2\nonumber\\
	\le&~\|(\Xb\Sb)^{\nicefrac{1}{2}}\Wb(\Ab\Db\Wb)^{\dagger}\Qb^{\nicefrac{1}{2}}\|_2\,\|\tilde{\fb}^{(t)}\|_2\label{eq:s1}.
\end{flalign}
In the above we used $\Qb^{\nicefrac{-1}{2}}(\Ab\Db^2\Ab^\ts\hat{\Delta\yb}-\pb)=\tilde{\fb}^{(t)}$. Using the SVD of $\Ab\Db\Wb$ and $\Qb$, we get $(\Ab\Db\Wb)^{\dagger}\Qb^{\nicefrac{1}{2}}=\widehat{\Vb}\Sigmab_\Qb^{-1/2}\Ub_\Qb^\ts\,\Ub_\Qb\Sigmab_\Qb^{1/2}\Ub_\Qb^\ts=\widehat{\Vb}\Ub_\Qb^\ts$. Now, note that $\Ub_\Qb\in\RR{m}{m}$ is an orthogonal matrix and $\|\widehat{\Vb}\|_2=1$. Therefore, combining with eqn.~\eqref{eq:s1} yields
	\begin{flalign}
	\|\vb\|_2\le&~\|(\Xb\Sb)^{\nicefrac{1}{2}}\Wb\widehat{\Vb}\Ub_\Qb^\ts\|_2\|\tilde{\fb}^{(t)}\|_2=\|(\Xb\Sb)^{\nicefrac{1}{2}}\Wb\widehat{\Vb}\|_2\|\tilde{\fb}^{(t)}\|_2\nonumber\\
	\le&\|(\Xb\Sb)^{\nicefrac{1}{2}}\Wb\|_2\|\tilde{\fb}^{(t)}\|_2.\label{eq:semi}
	\end{flalign}
The first equality follows from the unitary invariance property of the spectral norm and the second inequality follows from the sub-multiplicativity of the spectral norm and $\|\widehat{\Vb}\|_2=1$.
Our construction for $\Wb$ implies that eqn.~\eqref{eqn:pdprec} holds for any matrix $\Zb$ and, in particular, for $\Zb=(\Xb\Sb)^{\nicefrac{1}{2}}$. Eqn.~\eqref{eqn:pdprec} implies that
	\begin{flalign}\label{eq:w2}
	\nbr{(\Xb\Sb)^{\nicefrac{1}{2}} \Wb \Wb^\ts(\Xb\Sb)^{\nicefrac{1}{2}} - (\Xb\Sb)}_2\le \frac{\zeta}{4} \left(\|(\Xb\Sb)^{\nicefrac{1}{2}}\|_2^2+\frac{\|(\Xb\Sb)^{\nicefrac{1}{2}}\|_F^2}{m}\right)
	\end{flalign}
	holds with probability at least $1-\delta$. Applying Weyl's inequality on the left hand side of the  eqn.~\eqref{eq:w2}, we get
	\begin{flalign}\label{eq:semi2}
	\abs{\nbr{(\Xb\Sb)^{\nicefrac{1}{2}} \Wb}_2^2-\nbr{(\Xb\Sb)^{\nicefrac{1}{2}}}_2^2}\le  \frac{\zeta}{4} \left(\|(\Xb\Sb)^{\nicefrac{1}{2}}\|_2^2+\frac{\|(\Xb\Sb)^{\nicefrac{1}{2}}\|_F^2}{m}\right).
	\end{flalign}
	Using $\zeta\le 1$ and $\|(\Xb\Sb)^{\nicefrac{1}{2}}\|_2\le \|(\Xb\Sb)^{\nicefrac{1}{2}}\|_F \le \xb^\ts\sbb=n\mu$, we get\footnote{The constant 3 in eqn.~(\ref{eq:semi3}) could be slightly improved to \nicefrac{3}{2}; we chose to keep the suboptimal constant to make our results compatible with the work in~\cite{Mon03} and avoid reproving theorems from~\cite{Mon03}. The better constant does not result in any significant improvements in the number of iterations of the Algorithm~\ref{algo:iipm}.}
	\begin{flalign}\label{eq:semi3}
	\nbr{(\Xb\Sb)^{\nicefrac{1}{2}} \Wb}_2^2\le 3\|(\Xb\Sb)^{\nicefrac{1}{2}}\|_F^2=3n \mu.
	\end{flalign}
	Finally, combining eqns.~\eqref{eq:semi} and \eqref{eq:semi3}, we conclude
	\begin{flalign*}
	\|\vb\|_2\le\sqrt{3n\mu}\|\tilde{\fb}^{(t)}\|_2.
	\end{flalign*}
\end{proof}
\noindent The above result is particularly useful in proving the convergence of Algorithm~\ref{algo:iipm}. More precisely, combining a result from~\cite{Mon03} with our preconditioner $\Qb^{\nicefrac{-1}{2}}$, we can prove that
$\|\Qb^{\nicefrac{-1}{2}}\pb\|_2\le \Ocal(n)\sqrt{\mu}$.
This bound allows us to prove that if we run Algorithm~\ref{algo:PCG} for $\Ocal(\log n)$ iterations, then
$$\|\tilde{\fb}^{(t)}\|_2\le\frac{\gamma\sigma}{4\sqrt{\dimtwo}}\sqrt{\mu}\ \mbox{and}\ \|\vb\|_2\le\frac{\gamma\sigma}{4}\mu.$$
The last two inequalities are critical in the convergence analysis of Algorithm~\ref{algo:iipm}; see Appendix~\ref{app:convergence} for details.

We are now ready to present the infeasible IPM algorithm. We will need the following definition for the neighborhood
{\small\begin{flalign}
		&\mathcal{N}(\gamma)=\Big\{(\xb^k,\yb^k,\sbb^k):x_i^k s_i^k\ge(1-\gamma)\mu\,~\text{and}~\,\frac{\|\rb^k\|_2}{\|\rb^0\|_2} \leq \frac{\mu_k}{\mu_0},\Big\}.\nonumber
\end{flalign}}
Here $\gamma \in (0,1)$ and we note that the duality measure $\mu_k$ steadily reduces at each iteration.

\begin{algorithm}[H]
	\caption{Infeasible IPM}\label{algo:iipm}%
	\begin{algorithmic}
		
		\State \textbf{Input:}
		$\Ab\in\RR{\dimone}{\dimtwo}$,  $\bb\in\R{\dimone}$, $\cbb\in\R{\dimtwo}$, $\gamma \in (0,1)$, tolerance $\epsilon> 0$, centering parameter $\sig\in (0,\nicefrac{4}{5})$;
		
		\vspace{1mm}
		\State \textbf{Initialize:} $k\gets 0$; initial point $(\xb^{0},\yb^{0},\sbb^{0})$;
		
		\vspace{1mm}
		\While{$\mu_k > \epsilon$}
		
		\vspace{1mm}
		\State (a)~Compute sketching matrix $\Wb \in \mathbb{R}^{n \times w}$ (Section~\ref{sxn:background}) with $\zeta=1/2$ and $\delta = O(n^{-2})$;
		\vspace{1mm}
		\State (b)~Compute $\rb_p^k=\Ab\xb^k-\bb$; $\rb_d^k=\Ab^\ts\yb^k+\sbb^k-\cbb$; and $\pb^k$ from eqn.~(\ref{eqn:pdef});
		\vspace{1mm}
		\State (c)~Solve the linear system of eqn.~\eqref{eq:precond_alt} for $\zb$ using Algorithm~\ref{algo:PCG} with $\Wb$ from step (a) and $t=\Ocal(\log n)$. Compute $\hat{\Delta \yb}=\Qb^{\nicefrac{-1}{2}}{\zb}$;
		\vspace{1mm}
		\State (d)~Compute $\vb$ using eqn.~\eqref{eq:compv} with $\Wb$ from step (a); $\hat{\Delta\sbb}$ using eqn.~\eqref{eq:dels}; $\hat{\Delta\xb}$ using eqn.~\eqref{eq:delxhat};
		\vspace{1mm}
		\State (e)~\label{stepalpha1} Compute $\alphamax = \argmax\{ \alpha \in [0,1] : (\xb^k,\yb^k,\sbb^k) + \alpha (\hat{\Delta \xb}^k,\hat{\Delta \yb}^k,\hat{\Delta \sbb}^k)  \in \neigh\}$.
		\vspace{1mm}
		\State (f)~\label{stepalpha2}Compute $\alphaused = \argmin\{\alpha \in [0, \alphamax]: (\xb^k + \alpha \hat{\Delta \xb}^k)^\ts (\sbb^k + \alpha \hat{\Delta\sbb}^k)\}$.
		\vspace{1mm}
		\State (g)~Compute $(\xb^{k+1}, \yb^{k+1}, \sbb^{k+1}) = (\xb^k,\yb^k,\sbb^k) + \alphabar (\hat{\Delta \xb}^k,\hat{\Delta \yb}^k,\hat{\Delta \sbb}^k)$; set $k \gets k + 1$;
		
		\EndWhile
	\end{algorithmic}
\end{algorithm}

\vspace{0.02in}\noindent\textbf{Running time.} We start by discussing the running time to compute $\vb$. As discussed in Section~\ref{sxn:PCG},
$(\Ab\Db\Wb)^{\dagger}$ can be computed in $\Ocal(\nnz{\Ab}\cdot \log(m/\delta)+m^3\log(m/\delta))$ time.  Now, as $\Wb$ has $\Ocal(\log(m/\delta))$ non-zero entries per row, pre-multiplying by $\Wb$ takes $\Ocal(\nnz{\Ab}\log(m/\delta))$ time (assuming $\nnz{A}\ge n$). Since $\Xb$ and $\Sb$ are diagonal matrices, computing $\vb$ takes  $\Ocal(\nnz{\Ab}\cdot \log(m/\delta)+m^3\log(m/\delta))$ time, which is asymptotically the same as computing $\Qb^{\nicefrac{-1}{2}}$ (see eqn.~(\ref{eqn:svdQ})).

We now discuss the overall running time of Algorithm~\ref{algo:iipm}. At each iteration, with failure probability $\delta$, the preconditioner $\Qb^{\nicefrac{-1}{2}}$ and the vector $\vb$ can be computed in
$\Ocal(\nnz{\Ab}\cdot \log(m/\delta)+m^3\log(m/\delta))$ time. In addition, for $t=\Ocal(\log n)$ iterations of Algorithm~\ref{algo:PCG}, all the matrix-vector products in the CG solver can be computed in $\Ocal(\nnz{\Ab}\cdot \log n)$ time. Therefore, the computational time for steps (a)-(d) is given by $\Ocal(\nnz{\Ab}\cdot(\log n+ \log(m/\delta))+m^3\log(m/\delta))$. Finally, taking a union bound over all iterations with $\delta=\Ocal(n^{-2})$ (ignoring constant factors), Algorithm~\ref{algo:iipm} converges with probability at least 0.9. The running time at each iteration is given by
$\Ocal((\nnz{\Ab}+m^3)\log n)$.


\section{Extensions}\label{sxn:extensions}

We briefly discuss extensions of our work. First, there is nothing special about using a CG solver for solving eqn.~\eqref{eq:precond_alt}. We analyze two more solvers that could replace the proposed CG solver without
any loss in accuracy or any increase in the number of iterations for the long-step infeasible IPM Algorithm~\ref{algo:iipm} of Section~\ref{sxn:IIPM}. In Appendix~\ref{sxn:rchardson}, we analyze the performance of the
preconditioned Richardson Iteration and in Appendix~\ref{sxn:psd}, we analyze the performance of the preconditioned Steepest Descent.
In both cases, if the respective preconditioned solver (with the preconditioner of Section~\ref{sxn:PCG}) runs for $t = \Ocal(\log n)$ steps,  Theorem~\ref{thm:1} still holds, with small differences in the constant terms. While preconditioned Richardson iteration and preconditioned Steepest Descent are interesting from a theoretical perspective, they are not particularly practical.
%
In future work, we will also consider the preconditioned
Chebyshev semi-iterative method, which offers practical advantages compared to PCG in parallel settings.

Second, recall that our approach focused on full rank input matrices $\Ab \in \mathbb{R}^{m \times n}$ with $m \ll n$. Our overall approach still works if $\Ab$ in any $m \times n$ matrix that is low-rank, e.g., $\rank(\Ab)=k\ll \min\{\dimone,\dimtwo\}$. In that case, using the thin SVD of $\Ab$, we can rewrite the linear constraints as follows
%
$\Ub_\Ab\Sigmab_\Ab\Vb_\Ab^\ts\xb=\bb$,
%
where $\Ub_\Ab\in\RR{\dimone}{k}$ and $\Vb_\Ab\in\RR{\dimtwo}{k}$ are the matrices of left and right singular vecors of $\Ab$ respectively; $\Sigmab_\Ab\in\RR{k}{k}$ is the diagonal matrix with the $k$ non-zero singular values of $\Ab$ as its diagonal elements. The LP of eqn.~\eqref{eq:primal} can be restated as
\begin{flalign}
\min\,\cbb^\ts\xb\,,\text{ subject to }\Vb_\Ab^\ts\xb=\widetilde{\bb}\,,\xb\ge \zero\,,\label{eq:primal2}
\end{flalign}
where $\widetilde{\bb}=\Sigmab_\Ab^{-1}\Ub_\Ab^\ts\bb$. Note that, $\rank(\Vb_\Ab)=k\ll \dimtwo$ and therefore eqn.~\eqref{eq:primal2} can be solved using our framework. The matrices $\Ub_\Ab$, $\Vb_\Ab$, and $\Sigmab_\Ab$ can be approximately recovered using the fast SVD algorithms of \cite{Halko2011,BouDriMag14,clarkson2017low}. However, the accuracy of the final solution will depend on the accuracy of the approximate SVD and we defer this analysis to future work.

Third, even though we chose to use the Count-Min sketch and its analysis from~\cite{Cohen2016} (Section~\ref{sxn:background}), there are many other alternative sketching matrix constructions that would lead to similar results. A particularly simple one is the Gaussian sketching matrix $\Wb_G \in \mathbb{R}^{n \times w}$, where every entry is a $\mathcal{N}(0,1)$ random variable. Setting $w=\Ocal\left(\nicefrac{m+\log (1/ \delta)}{\zeta^{2}}\right)$
would result in the same accuracy guarantees as the sketching matrix of Section~\ref{sxn:background}. However, the (theoretical) running time needed to compute $\Ab \Db \Wb$ increases to $\Ocal (m \cdot\nnz{\Ab} )$. In practice, at least for relatively small matrices, using Gaussian sketching matrices is a reasonable alternative; see the discussion in~\cite{Meng2014SISC} which argued that the Gaussian matrix sketching-based solvers are considerably better than direct solvers. We also opted to use Gaussian matrices in our empirical evaluation, since we primarily interested in measuring the accuracy of the final solution as a function of the number of iterations of the solver and the IPM algorithm. Other known constructions of sketching matrices that are also applicable in our setting include (any) sub-gaussian sketching matrix; the Subsampled Randomized Hadamard transform (SRHT); and any of the Sparse Subspace Embeddings of~\cite{ClaWoo13,nelson2013osnap,meng2013low,cohen2016nearly}.
%


\begin{figure}[t]
	\centering
	\subfigure{
		\includegraphics[width=2.8in]{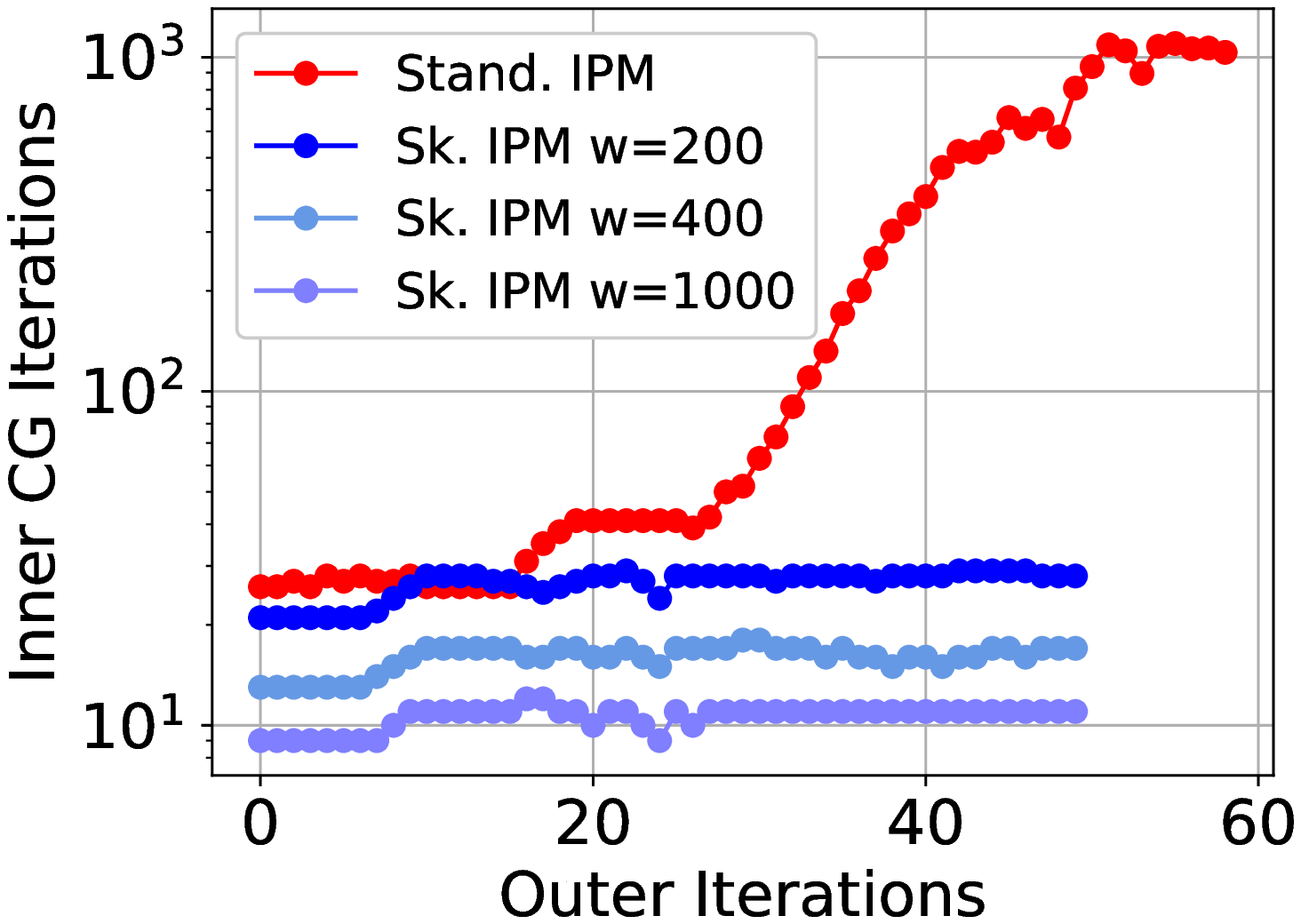}}
	\subfigure{
		\includegraphics[width=2.8in]{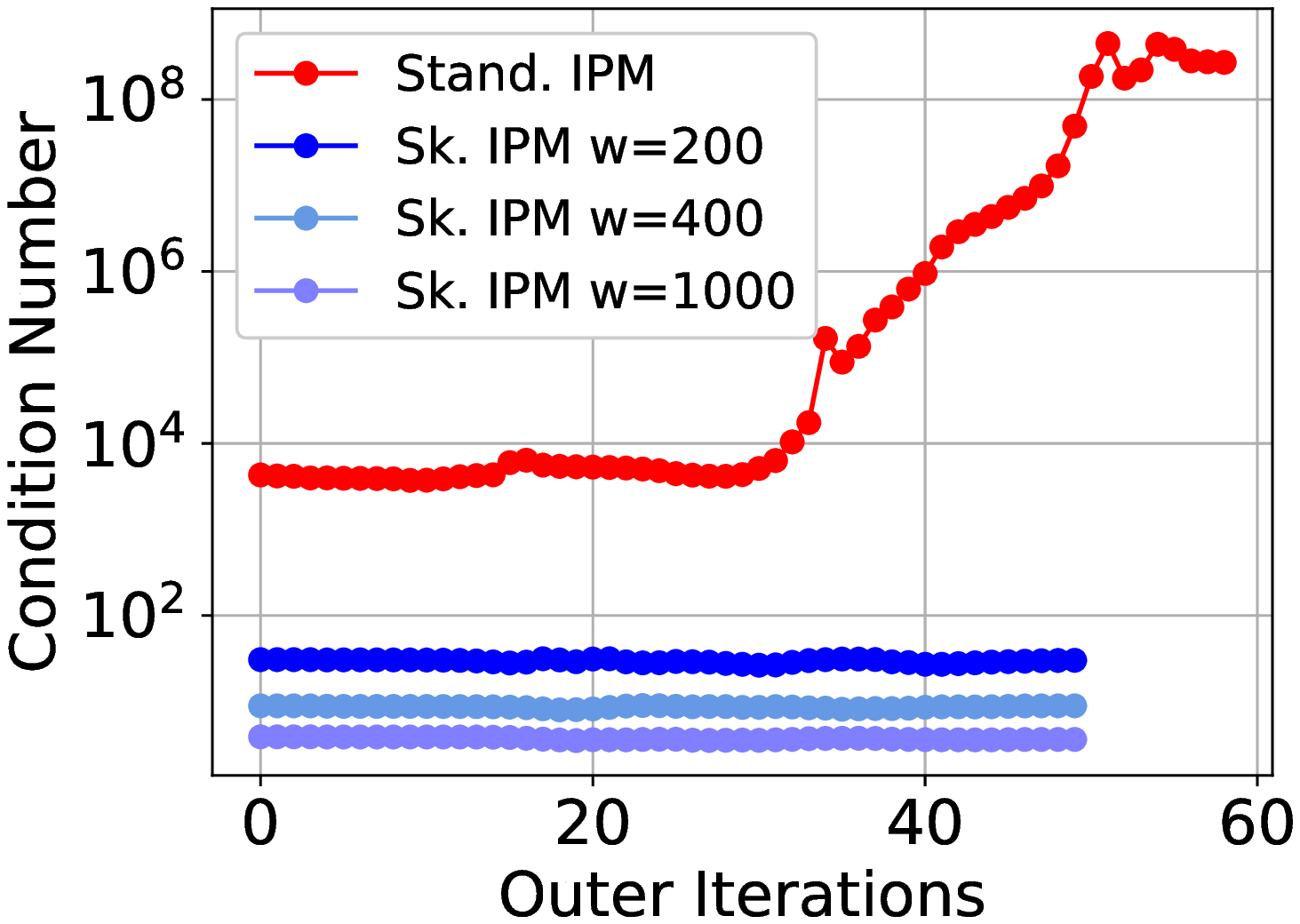}}
\vspace{-5pt}
\\
\addtocounter{subfigure}{-2}
\subfigure[ ]{
		\includegraphics[width=2.8in]{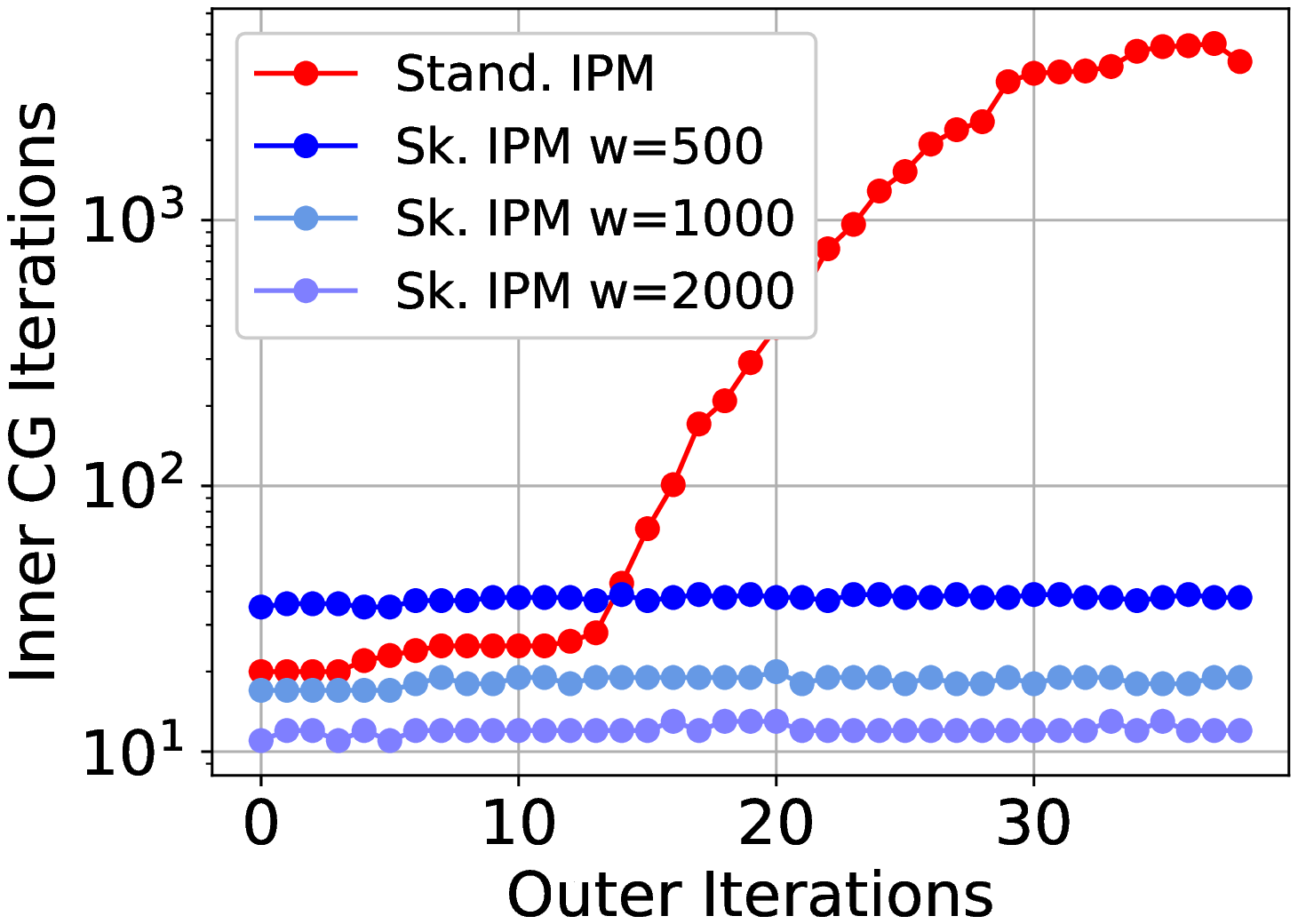}}
	\subfigure[]{
		\includegraphics[width=2.8in]{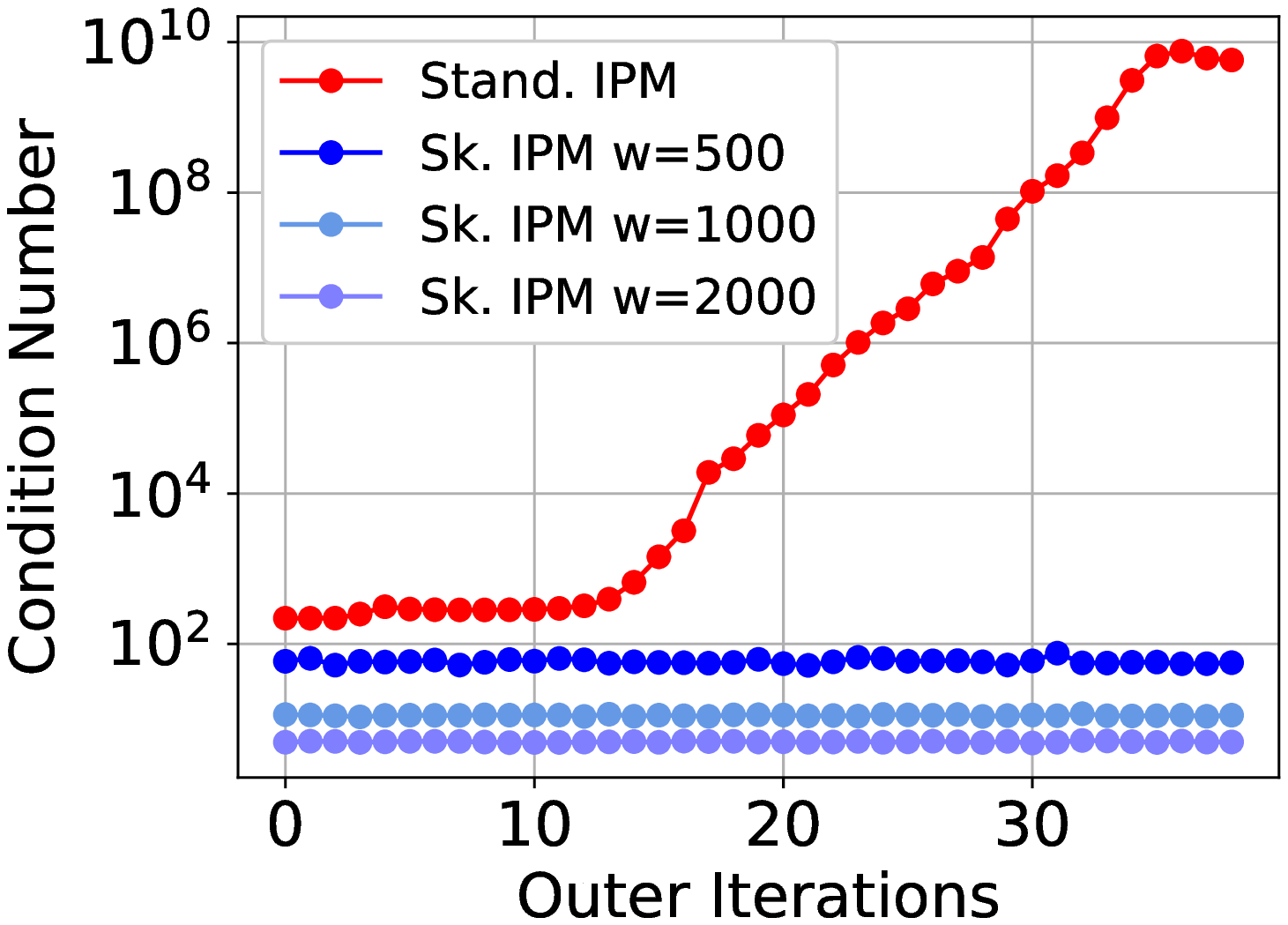}}
	\caption{\emph{ARCENE (top row) and DEXTER (bottom row) data sets}: Our algorithm (Sk. IPM) requires an order of magnitude fewer inner iterations than the Standard IPM with CG at each outer iteration, as demonstrated in (a). This is possibly due to the improved conditioning of \precNormal~compared to $\Ab \Db^2 \Ab^T$, as shown in (b).
		For all experiments \tolCG~$ = 10^{-5}$ and \tolOuterRes~$ = 10^{-9}$.}
	\label{fig:iter_ARCENE}

\end{figure}

\begin{figure}[t]
	\centering
	\subfigure[Max. Inner CG Iterations.]{
		\includegraphics[width=2.57in]{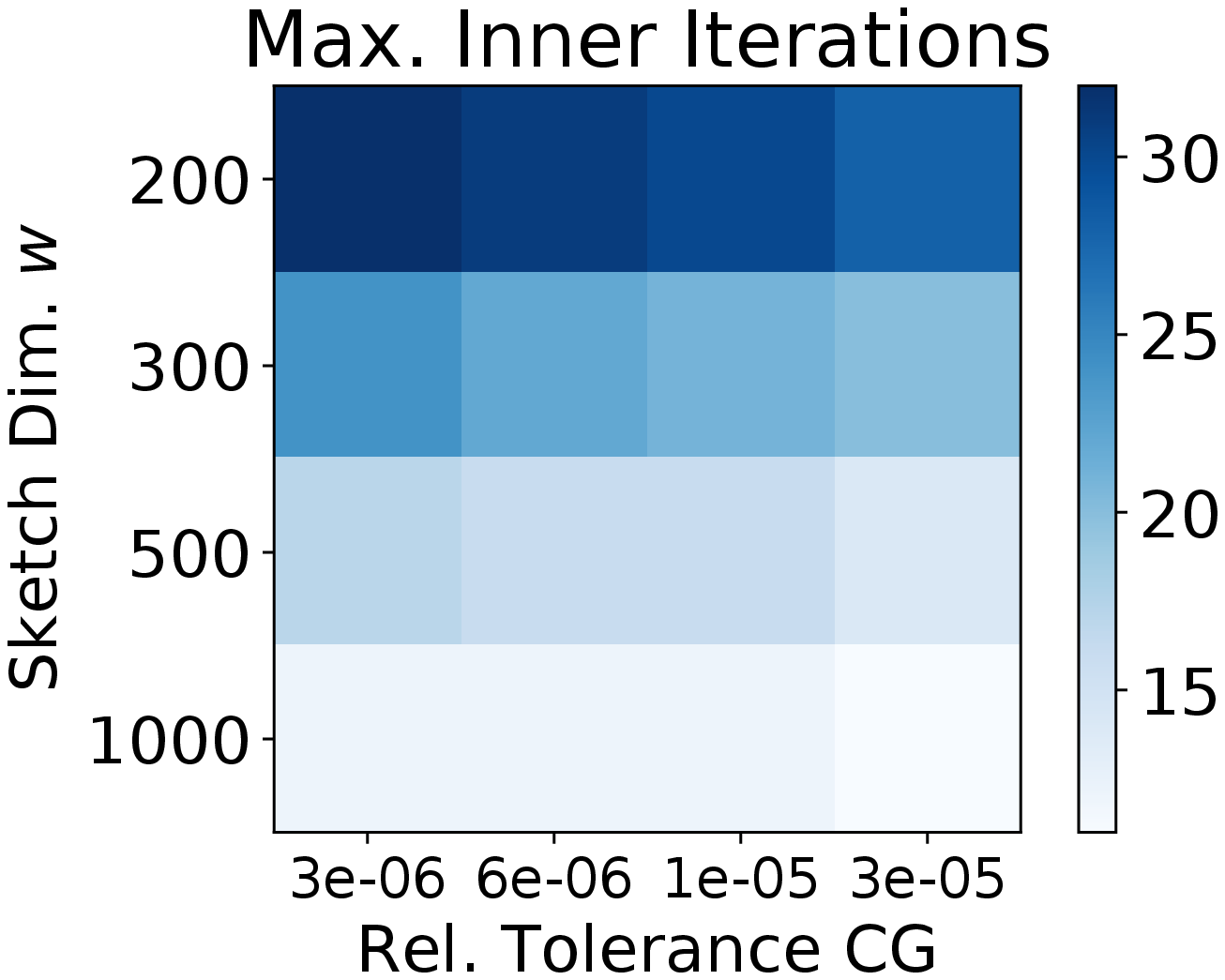}}
	\subfigure[Max. Condition Number.]{
		\label{fig:ARCENE_iter} 
		\includegraphics[width=2.57in]{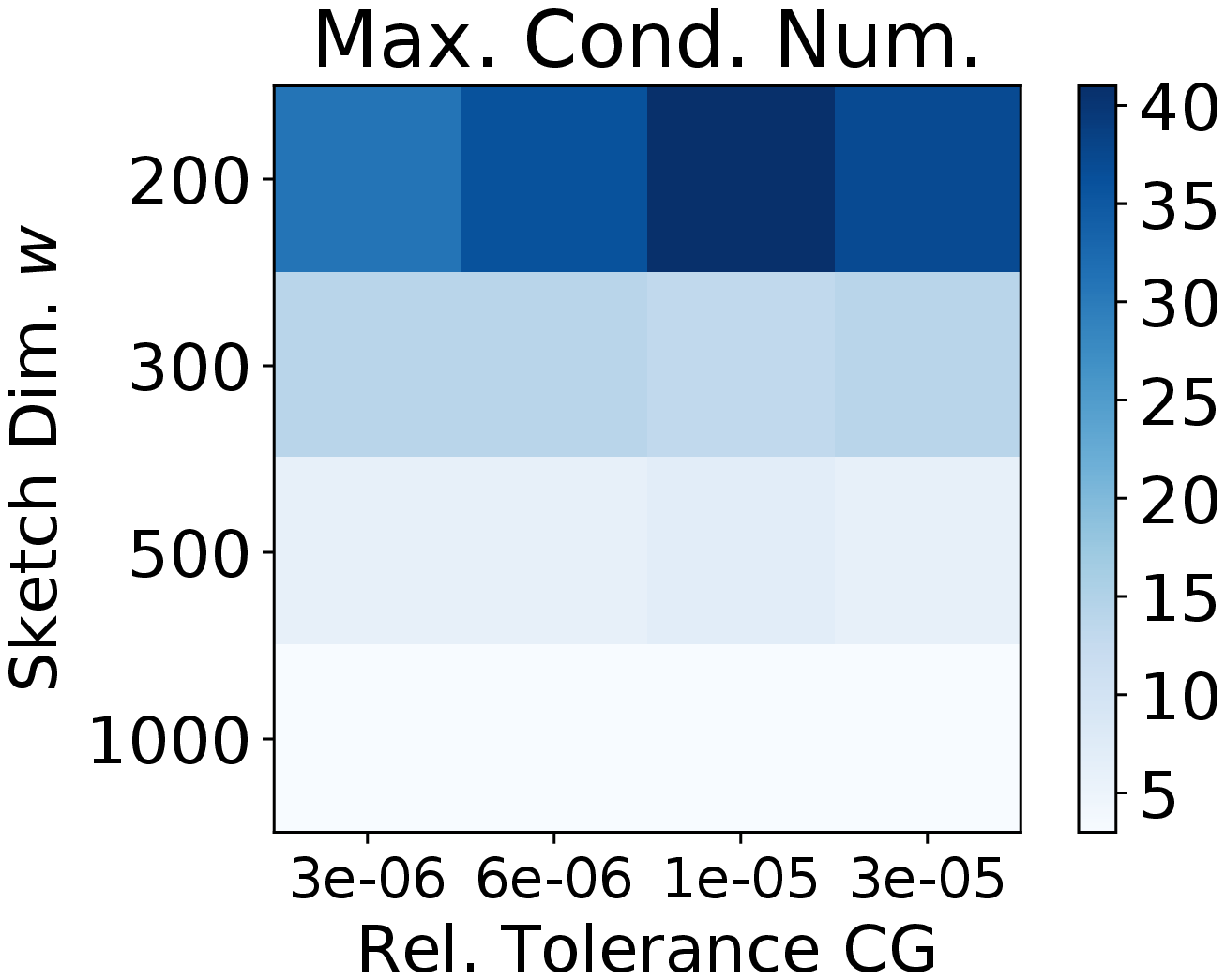}}
	\caption{\emph{ARCENE data set}: for various ($w$, \tolCG) settings,
		(a)~the maximum number of inner iterations used by our algorithm and (b)~the maximum condition number of \precNormal, across outer iterations. The standard IPM, across all settings, needed on the order of 1,000 iterations and  $\kappa(\Ab \Db^2\Ab^T)$ was on the order of $10^{8}$.
		The relative error was fixed to $0.04\%$. }
	\label{fig:ARCENE} 
\end{figure}

\section{Experiments}\label{sec:exp}

Here we demonstrate the empirical performance of our algorithm on a variety of real-world data sets from the UCI ML Repository~\cite{Dua2019}, such as ARCENE, DEXTER~\cite{guyon2005result}, DrivFace~\cite{diaz2016reduced},
and a gene expression cancer RNA-Sequencing dataset that is part of the  PANCAN dataset~\cite{Weinstein2013}. See Table~\ref{tableSVM} for a description of the datasets.
The experiments were implemented in Python and we observed that the results for both synthetic data (generated as described in Appendix~\ref{app:rand}) and real-world data were qualitatively similar. 
Thus, we highlight results on several representative real datasets. 

As an application, we consider $\ell_1$-regularized SVMs. All of the data sets are concerned with binary classification with $m \ll n$, where $n$ is the number of features.
In Appendix~\ref{app:svm}, we describe the $\ell_1$-SVM problem and how it can be formulated as an LP. Here, $m$ is the number of training points, $n$ is the feature dimension, and the size of the constraint matrix in the LP becomes $m \times (2n +1)$.
%


\vspace{0.02in}\noindent\textbf{Comparisons and Metrics}.
We compare our Algorithm~\ref{algo:iipm} with a standard IPM (see Chapter 10,~\cite{NumericalRecipes}) using CG, and a standard IPM using a direct solver.
We also use CVXPY as a benchmark to compare the accuracy of the solutions; we define the \emph{relative error} $\nicefrac{\|\hat{\xb} - \xb^\star\|_2}{\|\xb^\star\|_2}$, where $\hat{\xb}$ is our solution and $\xb^\star$ is the solution generated by CVXPY. We also consider the number of \emph{outer iterations}, namely the number of iterations of the IIPM algorithm, as well as the  number of \emph{inner iterations}, namely the number of iterations of the CG solver.
We denote the relative stopping tolerance for CG by \tolCG~ and we denote the
outer iteration residual error by \tolOuterRes. If not specified: \tolOuterRes~$= 10^{-9}$, \tolCG~$ = 10^{-5}$, and $\sigma = 0.5$.
We evaluated a Gaussian sketching matrix, and the initial triplet $(\xb, \yb, \sbb)$ for all IPM algorithms was set to be all ones.

\begin{table*}[ht]
	\caption{Comparison of (our) sketched IPM with CG, standard IPM with CG, and Standard IPM with a direct solver, for the $\ell_1$-SVM problem on UCI Machine Learning Repository~\cite{Dua2019} data sets. Across all, $\tau = 10^{-9}$ and a relative error of $10^{-3}$ or less was achieved. We define $\kappa_{\text{Sk}} = $ \kprecNormal~and $\kappa_{\text{Stan}} = \kappa(\Ab \Db^2\Ab^T)$.} \label{tableSVM}
	\begin{center}
		\resizebox{\textwidth}{!}{\begin{tabular}{|c|c|c|c|c|c|c|c|c|c|}
			\hline
			\multicolumn{1}{|c|}{\textbf{Problem}} &
			\multicolumn{1}{|c|}{\textbf{Size}} &
			\multicolumn{4}{|c|}{\textbf{Sketch IPM w/ Precond. CG}} &
			\multicolumn{3}{|c|}{\textbf{Stand. IPM w/ Unprec. CG}} &
			\multicolumn{1}{|c|}{\textbf{IPM w/ Dir.}} \\
			& \multicolumn{1}{c}{$(m \times N )$} & \multicolumn{1}{|c}{$w$} & \multicolumn{1}{c}{In. It.}  & \multicolumn{1}{c}{Out. It.} & \multicolumn{1}{c|}{$\kappa_{\text{Sk}}$}
			& \multicolumn{1}{c}{In. It.}  & \multicolumn{1}{c}{Out. It.} & \multicolumn{1}{c|}{$\kappa_{\text{Stan}}$}  & \multicolumn{1}{c|}{Out. It.}\\
			\hline
			ARCENE &$(100 \times 10K)$ & $200$ & $\mathbf{30}$ & $50$ & $38.09$
			& $\mathbf{1.1 K}$ & $59$ & $4.4 \times 10^{8}$ & $50$\\
			DEXTER &$(300 \times 20K)$ & $500$ & $\mathbf{39}$ & $39$ & $75.42$
			& $\mathbf{4.6K}$ & $39$ & $7.6 \times 10^{9}$ & $39$\\
			
			DrivFace &$(606 \times 6400)$ & $1000$ & $\mathbf{50}$ & $42$ & $68.87$
			& $\mathbf{139K}$ & $43$ & $17 \times 10^{12}$ & $42$\\
			Gene RNA &$(801 \times 20531)$ & $2000$ & $\mathbf{27}$ & $44$ & $20.03$
			& $\mathbf{101K}$ & $208$ & $4.7 \times 10^{12}$ & $44$\\
			\hline
		\end{tabular}}
	\end{center}
	\label{tab:multicol}
\end{table*}



\vspace{0.02in}\noindent\textbf{Experimental Results}.
Figure~\ref{fig:iter_ARCENE}(a) shows that our Algorithm~\ref{algo:iipm} uses an order of magnitude fewer \textit{inner} iterations than the un-preconditioned standard solver. This is due to the improved conditioning of the respective matrices in the normal equations, as demonstrated in Figure~\ref{fig:iter_ARCENE}(b).  Across various real-world and synthetic data sets, the results were qualitatively similar to those shown in Figure~\ref{fig:iter_ARCENE}. Results for several real-world data sets are summarized in Table~\ref{tableSVM}.

In general, our preconditioned CG solver used in Algorithm~\ref{algo:iipm} does not increase the total number of \textit{outer} iterations as compared to the standard IPM with CG, and the standard IPM with a direct linear solver (denoted IPM w/Dir), as seen in Table~\ref{tableSVM}. Actually, for unpreconditioned CG there is clearly more outer iterations, especially for Gene RNA, which has x5 outer iterations.
Figure~\ref{fig:iter_ARCENE} also demonstrates the relative insensitivity to the choice of $w$ (the sketching dimension, i.e., the number of columns of the sketching matrix $\Wb$ of Section~\ref{sxn:background}). For smaller values of $w$, our algorithm requires more inner iterations. However, across various choices of $w$, the number of inner iterations is always an order of magnitude smaller than the number required by the standard solver.

Figure~\ref{fig:ARCENE} shows the performance of our algorithm for a range of ($w$, \tolCG) pairs.
Figure~\ref{fig:ARCENE}(a) demonstrates that the number of the inner iterations is robust to the choice of \tolCG~ and $w$. The number of inner iterations varies between $15$ and $35$ for the ARCENE data set, while the standard IPM took on the order of $1,000$ iterations across all parameter settings. Across all settings, the relative error was fixed at $0.04\%$. In general, our sketched IPM is able to produce an extremely high accuracy solution across parameter settings. Thus we do not report additional numerical results for the relative error, which was consistently $10^{-3}$ or less.
Figure~\ref{fig:ARCENE}(b) demonstrates a trade-off of our approach: as both \tolCG~ and $w$ are increased, the condition number \kprecNormal~decreases, corresponding to better conditioned systems. As a result,  fewer inner iterations are required. 
In this context, Figure~\ref{fig:ARCENE_more} shows that how the number of inner CG iterations (Figure~\ref{fig:ARCENE_v_more1}) or the condition number of $\Qb^{-1/2}\Ab\Db^2\Ab^\ts\Qb^{-1/2}$ (Figure~\ref{fig:ARCENE_v_more2}) decreases with the increase in sketching dimension $w$ for various \tolCG\,.

\begin{figure}[t]
	\centering
	\subfigure[ ]{
		\label{fig:ARCENE_v_more1} 
		\includegraphics[width=2.8in]{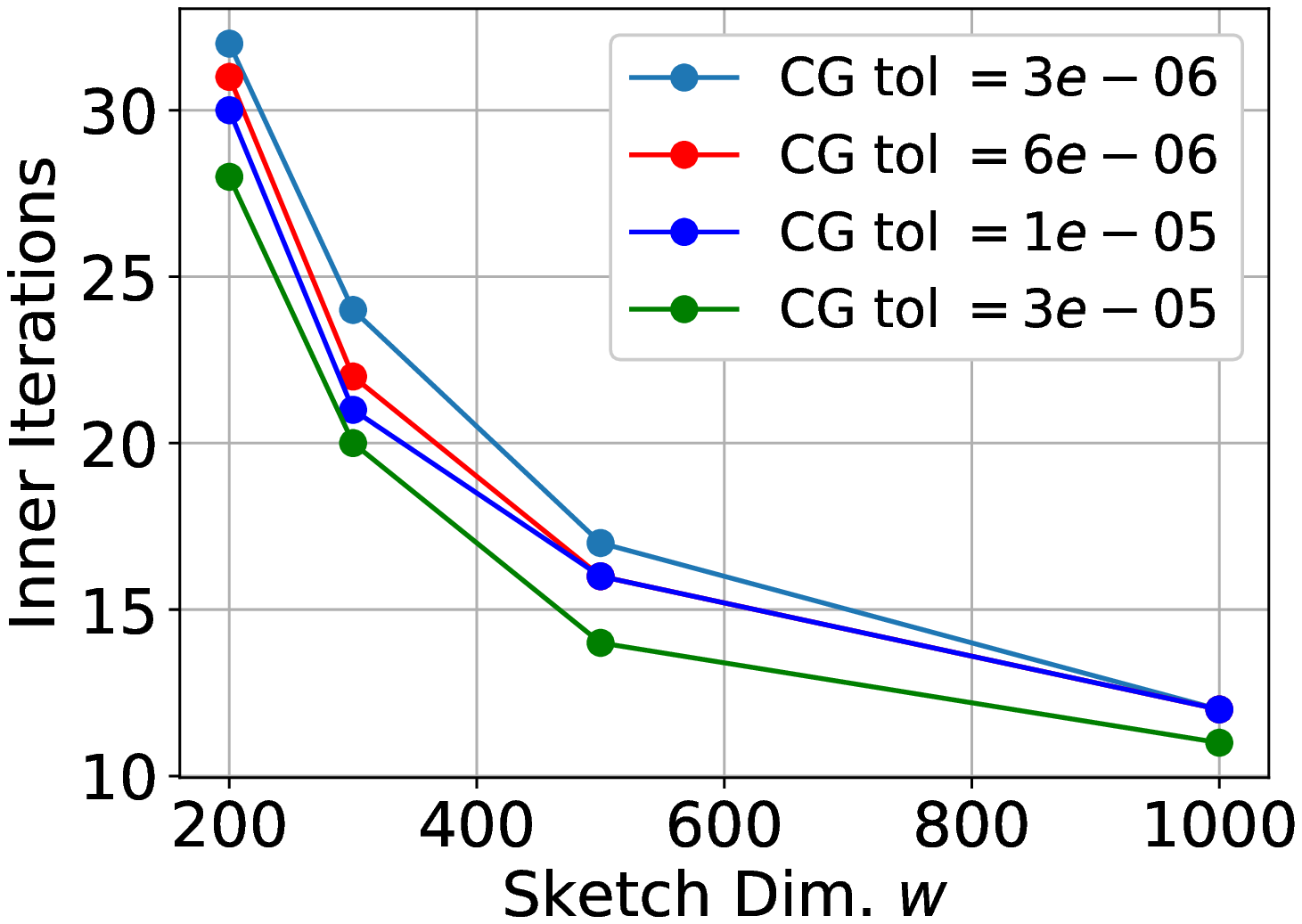}}
	\subfigure[]{
		\label{fig:ARCENE_v_more2} 
		\includegraphics[width=2.8in]{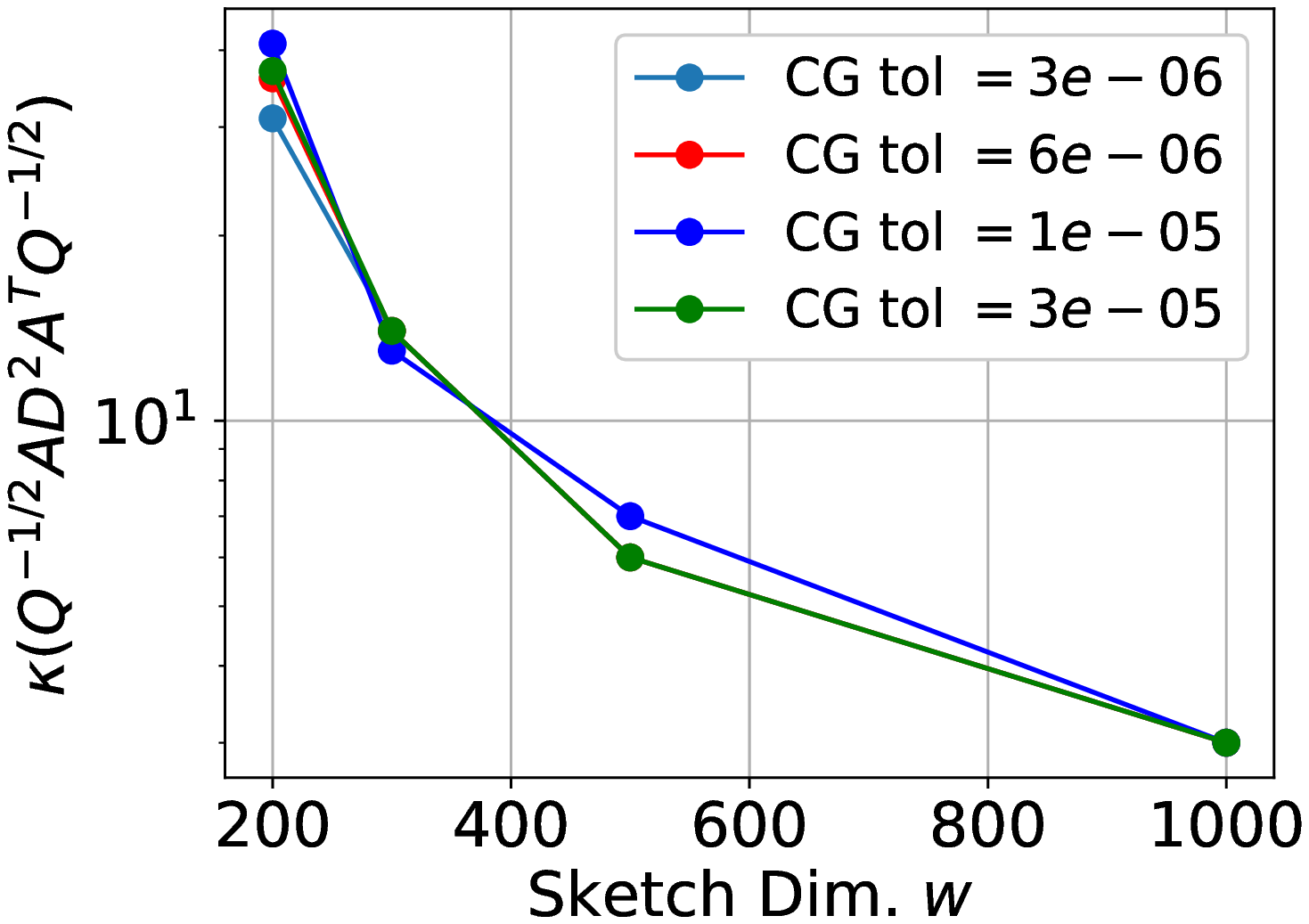}}
	\caption{\emph{ARCENE data set}: As $w$ increases, (a) the number of inner iterations decreases, and is relatively robust to \tolCG and (b) the condition number decreases as well.
	}
	\label{fig:ARCENE_more}
\end{figure}





\section{Conclusions}
We proposed and analyzed an infeasible IPM algorithm using a preconditioned conjugate gradient solver for the normal equations and a novel perturbation vector to correct for the error due to the approximate solver. Thus, we speed up each iteration of the IPM algorithm, without increasing the overall number of iterations. We demonstrate empirically that our IPM requires an order of magnitude fewer inner iterations within each linear solve than standard IPMs. It would be interesting to extend our work to analyze feasible IPMs. More precisely, we would like to apply Algorithm~\ref{algo:iipm} of Section~\ref{sxn:IIPM} starting with a strictly feasible point. In that case, the analysis should be simpler and the iteration complexity of the IPM algorithm should reduce to $\Ocal (n \log(1/\epsilon))$, which is the best known for feasible long-step path following IPM algorithms. We chose to present the more technically challenging approach in this paper and delegate the feasible case to future work.


\section*{Acknowledgements}

AC and PD were partially supported by NSF FRG 1760353 and NSF CCF BSF 1814041. HA was partially supported by BSF grant 2017698. PL was supported by an Amazon Graduate Fellowship in Artificial Intelligence.

\setlength{\bibsep}{6pt}
\bibliographystyle{plain}
{
\bibliography{bibliography}

\begin{thebibliography}{10}

\bibitem{Avron2010}
Haim Avron, Petar Maymounkov, and Sivan Toledo.
\newblock {B}lendenpik: {S}upercharging {LAPACK}'s least-squares solver.
\newblock {\em SIAM Journal on Scientific Computing}, 32(3):1217--1236, 2010.

\bibitem{axelsson1984finite}
Owe Axelsson and Vincent~A. Barker.
\newblock {\em Finite element solution of boundary value problems: {T}heory and
  computation}, volume~35.
\newblock Society for Industrial and Applied Mathematics, 1984.

\bibitem{Bienstock06}
Daniel Bienstock and Garud Iyengar.
\newblock Approximating fractional packings and coverings in ${\Ocal}
  (1/\epsilon)$ iterations.
\newblock {\em SIAM Journal on Computing}, 35(4):825--854, 2006.

\bibitem{BouDriMag14}
Christos Boutsidis, Petros Drineas, and Malik Magdon-Ismail.
\newblock Near-optimal column-based matrix reconstruction.
\newblock {\em SIAM Journal on Computing}, 43(2):687--717, 2014.

\bibitem{bouyouli2009new}
R.~Bouyouli, G{\'e}rard Meurant, Laurent Smoch, and Hassane Sadok.
\newblock New results on the convergence of the conjugate gradient method.
\newblock {\em Numerical Linear Algebra with Applications}, 16(3):223--236,
  2009.

\bibitem{brand2020solving}
Jan van~den Brand, Yin~Tat Lee, Aaron Sidford, and Zhao Song.
\newblock Solving tall dense linear programs in nearly linear time.
\newblock {\em arXiv preprint arXiv:2002.02304}, 2020.

\bibitem{CYD18}
Agniva Chowdhury, Jiasen Yang, and Petros Drineas.
\newblock An iterative, sketching-based framework for ridge regression.
\newblock In {\em Proceedings of the 35th International Conference on Machine
  Learning}, volume~80, pages 988--997, 2018.

\bibitem{ClaWoo13}
Kenneth~L. Clarkson and David~P. Woodruff.
\newblock Low rank approximation and regression in input sparsity time.
\newblock In {\em Proceedings of the 45th Annual ACM symposium on Theory of
  Computing}, pages 81--90, 2013.

\bibitem{clarkson2017low}
Kenneth~L Clarkson and David~P Woodruff.
\newblock Low-rank approximation and regression in input sparsity time.
\newblock {\em Journal of the ACM (JACM)}, 63(6):54, 2017.

\bibitem{cohen2016nearly}
Michael~B. Cohen.
\newblock Nearly tight oblivious subspace embeddings by trace inequalities.
\newblock In {\em Proceedings of the 27th Annual ACM-SIAM Symposium on Discrete
  Algorithms}, pages 278--287, 2016.

\bibitem{CLS19}
Michael~B. Cohen, Yin~Tat Lee, and Zhao Song.
\newblock Solving linear programs in the current matrix multiplication time.
\newblock In {\em Proceedings of the 51st Annual ACM Symposium on Theory of
  Computing}, pages 938--942, 2019.

\bibitem{Cohen2016}
Michael~B. Cohen, Jelani Nelson, and David~P. Woodruff.
\newblock Optimal approximate matrix product in terms of stable rank.
\newblock In {\em 43rd International Colloquium on Automata, Languages, and
  Programming}, pages 11:1--11:14, 2016.

\bibitem{CMTH16}
Yiran Cui, Keiichi Morikuni, Takashi Tsuchiya, and Ken Hayami.
\newblock Implementation of interior-point methods for {LP} based on krylov
  subspace iterative solvers with inner-iteration preconditioning.
\newblock {\em Computational Optimization and Applications}, 74(1):143--176,
  2019.

\bibitem{daitch2008faster}
Samuel~I. Daitch and Daniel~A. Spielman.
\newblock Faster approximate lossy generalized flow via interior point
  algorithms.
\newblock In {\em Proceedings of the 40th Annual ACM Symposium on Theory of
  Computing}, pages 451--460, 2008.

\bibitem{diaz2016reduced}
Katerine Diaz-Chito, Aura Hern{\'a}ndez-Sabat{\'e}, and Antonio~M. L{\'o}pez.
\newblock A reduced feature set for driver head pose estimation.
\newblock {\em Applied Soft Computing}, 45:98--107, 2016.

\bibitem{Donoho05}
David~L. Donoho and Jared Tanner.
\newblock Sparse nonnegative solution of underdetermined linear equations by
  linear programming.
\newblock In {\em Proceedings of the National Academy of Sciences of the United
  States of America}, pages 9446--9451, 2005.

\bibitem{Drineas2016}
Petros Drineas and Michael~W. Mahoney.
\newblock {RandNLA}: {R}andomized numerical linear algebra.
\newblock {\em Communications of the ACM}, 59(6):80--90, 2016.

\bibitem{DM2018}
Petros Drineas and Michael~W. Mahoney.
\newblock {\em Lectures on randomized numerical linear algebra}, volume~25 of
  {\em The Mathematics of Data, IAS/Park City Mathematics Series}.
\newblock American Mathematical Society, 2018.

\bibitem{Dua2019}
Dheeru Dua and Casey Graff.
\newblock {UCI} machine learning repository, 2017.

\bibitem{GVL12}
Gene~H. Golub and Charles~F. Van~Loan.
\newblock {\em Matrix computations}, volume~3.
\newblock Johns Hopkins University Press, 2012.

\bibitem{guyon2005result}
Isabelle Guyon, Steve Gunn, Asa Ben-Hur, and Gideon Dror.
\newblock Result analysis of the {NIPS} 2003 feature selection challenge.
\newblock In {\em Advances in Neural Information Processing Systems}, pages
  545--552, 2005.

\bibitem{Halko2011}
Nathan Halko, Per-Gunnar Martinsson, and Joel~A. Tropp.
\newblock Finding structure with randomness: Probabilistic algorithms for
  constructing approximate matrix decompositions.
\newblock {\em SIAM Review}, 53(2):217--288, 2011.

\bibitem{karmarkar84}
Narendra Karmarkar.
\newblock A new polynomial-time algorithm for linear programming.
\newblock In {\em Proceedings of the 16th Annual ACM Symposium on Theory of
  Computing}, pages 302--311, 1984.

\bibitem{LS13}
Yin~Tat Lee and Aaron Sidford.
\newblock Path finding {I}: Solving linear programs with
  $\tilde{\Ocal}(\sqrt{rank})$ linear system solves.
\newblock {\em arXiv preprint arXiv:1312.6677}, 2013.

\bibitem{lee2013path2}
Yin~Tat Lee and Aaron Sidford.
\newblock Path finding {II}: An $\tilde{\Ocal}(m\sqrt{n})$ algorithm for the
  minimum cost flow problem.
\newblock {\em arXiv preprint arXiv:1312.6713}, 2013.

\bibitem{LS14}
Yin~Tat Lee and Aaron Sidford.
\newblock Path finding methods for linear programming: {S}olving linear
  programs in $\tilde{\Ocal} (\sqrt{rank})$ iterations and faster algorithms
  for maximum flow.
\newblock In {\em Proceedings of the 55th IEEE Symposium on Foundations of
  Computer Science}, pages 424--433, 2014.

\bibitem{LS15}
Yin~Tat Lee and Aaron Sidford.
\newblock Efficient inverse maintenance and faster algorithms for linear
  programming.
\newblock In {\em Proceedings of the 56th IEEE Symposium on Foundations of
  Computer Science}, pages 230--249, 2015.

\bibitem{lee2019solving}
Yin~Tat Lee and Aaron Sidford.
\newblock Solving linear programs with $\tilde{\Ocal}(\sqrt{rank})$ linear
  system solves.
\newblock {\em arXiv preprint arXiv:1910.08033}, 2019.

\bibitem{LondonAAAI2018}
Palma London, Shai Vardi, Adam Wierman, and Hanling Yi.
\newblock A parallelizable acceleration framework for packing linear programs.
\newblock In {\em Proceedings of the 32nd AAAI Conference on Artificial
  Intelligence}, pages 3706 -- 3713, 2018.

\bibitem{Luenberger15}
David~G. Luenberger and Yinyu Ye.
\newblock {\em Linear and Nonlinear Programming}.
\newblock Springer Publishing Company, Incorporated, 3rd edition, 2008.

\bibitem{Mahoney11}
Michael~W. Mahoney.
\newblock Randomized algorithms for matrices and data.
\newblock {\em Foundations and Trends in Machine Learning}, 3(2):123--224,
  2011.

\bibitem{meng2013low}
Xiangrui Meng and Michael~W. Mahoney.
\newblock Low-distortion subspace embeddings in input-sparsity time and
  applications to robust linear regression.
\newblock In {\em Proceedings of the 45th Annual ACM Symposium on Theory of
  Computing}, pages 91--100, 2013.

\bibitem{Meng2014SISC}
Xiangrui Meng, Michael~A. Saunders, and Michael~W. Mahoney.
\newblock {LSRN: A parallel iterative solver for strongly over- or
  underdetermined systems}.
\newblock {\em SIAM Journal on Scientific Computing}, 36(2):95--118, 2014.

\bibitem{meshi2011alternating}
Ofer Meshi and Amir Globerson.
\newblock An alternating direction method for dual {MAP LP} relaxation.
\newblock In {\em Joint European Conference on Machine Learning and Knowledge
  Discovery in Databases}, pages 470--483. Springer, 2011.

\bibitem{Mon03}
Renato D.~C. Monteiro and Jerome~W. O’Neal.
\newblock Convergence analysis of a long-step primal-dual infeasible
  interior-point {LP} algorithm based on iterative linear solvers.
\newblock {\em Georgia Institute of Technology}, 2003.

\bibitem{Mon04}
Renato D.~C. Monteiro, Jerome~W. O'Neal, and Takashi Tsuchiya.
\newblock Uniform boundedness of a preconditioned normal matrix used in
  interior-point methods.
\newblock {\em SIAM Journal on Optimization}, 15(1):96--100, 2004.

\bibitem{nelson2013osnap}
Jelani Nelson and Huy~L. Nguy{\^e}n.
\newblock {OSNAP}: {F}aster numerical linear algebra algorithms via sparser
  subspace embeddings.
\newblock In {\em Proceedings of the 54th IEEE Symposium on Foundations of
  Computer Science}, pages 117--126, 2013.

\bibitem{NW06}
Jorge Nocedal and Stephen Wright.
\newblock {\em Numerical optimization}.
\newblock Springer Science \& Business Media, 2006.

\bibitem{paige1975solution}
Christopher~C. Paige and Michael~A. Saunders.
\newblock Solution of sparse indefinite systems of linear equations.
\newblock {\em SIAM Journal on Numerical Analysis}, 12(4):617--629, 1975.

\bibitem{PW17}
Mert Pilanci and Martin~J. Wainwright.
\newblock Newton sketch: A near linear-time optimization algorithm with
  linear-quadratic convergence.
\newblock {\em SIAM Journal on Optimization}, 27(1):205--245, 2017.

\bibitem{NumericalRecipes}
William~H. Press, Saul~A. Teukolsky, William~T. Vetterling, and Brian~P.
  Flannery.
\newblock Numerical recipes 3rd edition: {T}he art of scientific computing.
\newblock In {\em The Oxford Handbook of Innovation}, chapter~10. Cambridge
  University Press, 2007.

\bibitem{recht2012factoring}
Ben Recht, Christopher Re, Joel Tropp, and Victor Bittorf.
\newblock Factoring nonnegative matrices with linear programs.
\newblock In {\em Advances in Neural Information Processing Systems}, pages
  1214--1222, 2012.

\bibitem{RV93}
Mauricio G.~C. Resende and Geraldo Veiga.
\newblock An implementation of the dual affine scaling algorithm for
  minimum-cost flow on bipartite uncapacitated networks.
\newblock {\em SIAM Journal on Optimization}, 3(3):516--537, 1993.

\bibitem{spielman2004nearly}
Daniel~A. Spielman and Shang-Hua Teng.
\newblock Nearly-linear time algorithms for graph partitioning, graph
  sparsification, and solving linear systems.
\newblock In {\em Proceedings of the 36th Annual ACM Symposium on Theory of
  Computing}, volume~4, pages 81--90, 2004.

\bibitem{spielman2014nearly}
Daniel~A. Spielman and Shang-Hua Teng.
\newblock Nearly linear time algorithms for preconditioning and solving
  symmetric, diagonally dominant linear systems.
\newblock {\em SIAM Journal on Matrix Analysis and Applications},
  35(3):835--885, 2014.

\bibitem{VPL18}
Ky~Vu, Pierre-Louis Poirion, and Leo Liberti.
\newblock Random projections for linear programming.
\newblock {\em Mathematics of Operations Research}, 43(4):1051--1071, 2018.

\bibitem{Weinstein2013}
John~N. Weinstein, Eric~A. Collisson, Gordon~B. Mills, Kenna R.~Mills Shaw,
  Brad~A. Ozenberger, Kyle Ellrott, Ilya Shmulevich, Chris Sander, Joshua~M.
  Stuart, et~al.
\newblock The cancer genome atlas pan-cancer analysis project.
\newblock {\em Nature Genetics}, 45(10):1113--1120, 2013.

\bibitem{Woodruff14}
David~P. Woodruff.
\newblock Sketching as a tool for numerical linear algebra.
\newblock {\em Foundations and Trends in Theoretical Computer Science},
  10(1-2), 2014.

\bibitem{wright1997primal}
Stephen~J. Wright.
\newblock {\em Primal-dual interior-point methods}, volume~54.
\newblock Society for Industrial and Applied Mathematics, 1997.

\bibitem{XYRCM16}
Peng Xu, Jiyan Yang, Farbod Roosta-Khorasani, Christopher R{\'e}, and
  Michael~W. Mahoney.
\newblock Sub-sampled {Newton} methods with non-uniform sampling.
\newblock In {\em Advances in Neural Information Processing Systems}, pages
  3000--3008, 2016.

\bibitem{yang2011alternating}
Junfeng Yang and Yin Zhang.
\newblock Alternating direction algorithms for $\ell_1$-problems in compressive
  sensing.
\newblock {\em SIAM Journal on Scientific Computing}, 33(1):250--278, 2011.

\bibitem{yuan2010high}
Ming Yuan.
\newblock High dimensional inverse covariance matrix estimation via linear
  programming.
\newblock {\em Journal of Machine Learning Research}, 11(Aug):2261--2286, 2010.

\bibitem{Zh94}
Yin Zhang.
\newblock On the convergence of a class of infeasible interior-point methods
  for the horizontal linear complementarity problem.
\newblock {\em SIAM Journal on Optimization}, 4(1):208--227, 1994.

\bibitem{zhu20041}
Ji~Zhu, Saharon Rosset, Robert Tibshirani, and Trevor~J. Hastie.
\newblock 1-norm support vector machines.
\newblock In {\em Advances in Neural Information Processing Systems}, pages
  49--56, 2004.

\end{thebibliography}
}

\newpage

\begin{appendices}

\newpage


\section{Richardson Iteration}\label{sxn:rchardson}

Here, we show that all our analyses still hold if we replace Step 4 of Algorithm~\ref{algo:PCG} (CG solver) with Richardson's iteration. Basically, all we need to show is that the condition of eqn.~\eqref{eq:pdcond2} holds. Note that the condition of eqn.~\eqref{eq:pdcond1} already holds from Lemma~\ref{lem:cond3}, as we use the sketching matrix $\Wb\in\RR{n}{w}$ discussed in Section~\ref{sxn:PCG}.

\begin{algorithm}[H]
	\caption{Richardson Iteration Solver}\label{algo:iterative_solver}
	\begin{algorithmic}
		\State \textbf{Input:}
		$\Ab\Db\in\RR{\dimone}{\dimtwo}$, $\pb\in\R{\dimone}$;
		number of iterations $t>0$;
		sketching matrix $\Wb \in \mathbb{R}^{\dimtwo\times w}$;
		\State \textbf{Initialize:}
		$\tilde{\zb}^{0} \gets \zero_\dimone $;
		
		\For{$j=1$ \textbf{to} $t$}
		\State $\tilde{\zb}^{j} \gets \tilde{\zb}^{j-1}+\Qb^{-\nicefrac{1}{2}}(\pb-\Ab\Db^2\Ab^\ts\Qb^{-\nicefrac{1}{2}}\tilde{\zb}^{j-1})$;
		\EndFor
		
		\State \textbf{Output:} return $\tilde{\zb}^{t}$;
	\end{algorithmic}
\end{algorithm}

\noindent Our first result expresses the residual vector $\tilde{\fb}^{(j)}$ in terms of $\tilde{\fb}^{(j-1)}$ for $j=1\ldots t$.
\begin{lemma}\label{lem:recursive}
	Let $\tilde{\fb}^{(j)}$, $j=1\ldots t$ be the residual vectors at each iteration.Then,
	\begin{flalign}
	\tilde{\fb}^{(j)}=~\left(\Ib_n-\Qb^{-\nicefrac{1}{2}}\Ab\Db^2\Ab^\ts\Qb^{-\nicefrac{1}{2}}\right)\tilde{\fb}^{(j-1)}\label{eq:error_recursive2}\,.
	\end{flalign}
	Recall that $\Qb=\Ab\Db\Wb\Wb^\ts\Db\Ab^\ts$ and $\tilde{\fb}^{(j)}=\Qb^{-\nicefrac{1}{2}}(\Ab\Db^2\Ab^\ts\Qb^{-\nicefrac{1}{2}}\tilde{\zb}^j-\pb)$.
\end{lemma}
\begin{proof}
	Using Algorithm~\ref{algo:iterative_solver}, we express $\tilde{\fb}^{(j)}$ as
	\begin{flalign}
	\tilde{\fb}^{(j)}=&~\Qb^{-\nicefrac{1}{2}}\Ab\Db^2\Ab^\ts\Qb^{-\nicefrac{1}{2}}\tilde{\zb}^{j}-\Qb^{-\nicefrac{1}{2}}\pb\nonumber\\
	=&~\Qb^{-\nicefrac{1}{2}}\Ab\Db^2\Ab^\ts\Qb^{-\nicefrac{1}{2}}\left(\tilde{\zb}^{j-1}+\Qb^{-\nicefrac{1}{2}}(\pb-\Ab\Db^2\Ab^\ts\Qb^{-\nicefrac{1}{2}}\tilde{\zb}^{j-1})\right)-\Qb^{-\nicefrac{1}{2}}\pb\nonumber\\
	=&~\left(\Qb^{-\nicefrac{1}{2}}\Ab\Db^2\Ab^\ts\Qb^{-\nicefrac{1}{2}}\tilde{\zb}^{j-1}-\Qb^{-\nicefrac{1}{2}}\pb\right)\nonumber\\
	&~~~~~~~~~~~~~~~~~~~~~-\Qb^{-\nicefrac{1}{2}}\Ab\Db^2\Ab^\ts\Qb^{-\nicefrac{1}{2}}\left(\Qb^{-\nicefrac{1}{2}}\Ab\Db^2\Ab^\ts\Qb^{-\nicefrac{1}{2}}\tilde{\zb}^{j-1}-\Qb^{-\nicefrac{1}{2}}\pb\right)\nonumber\\
	=&~\left(\Ib_\dimone-\Qb^{-\nicefrac{1}{2}}\Ab\Db^2\Ab^\ts\Qb^{-\nicefrac{1}{2}}\right)\left(\Qb^{-\nicefrac{1}{2}}\Ab\Db^2\Ab^\ts\Qb^{-\nicefrac{1}{2}}\tilde{\zb}^{j-1}-\Qb^{-\nicefrac{1}{2}}\pb\right)\nonumber\\
	=&~\left(\Ib_\dimone-\Qb^{-\nicefrac{1}{2}}\Ab\Db^2\Ab^\ts\Qb^{-\nicefrac{1}{2}}\right)\tilde{\fb}^{(j-1)}\nonumber\,,
	\end{flalign}
	which concludes the proof.
\end{proof}

\noindent In the next result, we show that the spectral norm of $\Ib_\dimone-\Qb^{-\nicefrac{1}{2}}\Ab\Db^2\Ab^\ts\Qb^{-\nicefrac{1}{2}}$ is upper bounded by $\zeta$.
\begin{lemma}\label{lem:normbound}
	Let the condition of eqn.~\eqref{eq:pdcond1} holds for the sketching matrix $\Wb\in\RR{\dimtwo}{w}$. Then
	\begin{flalign}
	\|\Qb^{-\nicefrac{1}{2}}\Ab\Db^2\Ab^\ts\Qb^{-\nicefrac{1}{2}}-\Ib_\dimone\|_2\le\zeta\,.\nonumber
	\end{flalign}
\end{lemma}
\begin{proof}
	As the condition in eqn.~\eqref{eq:pdcond1} holds,
	we can go backwards in the proof of Lemma~\ref{lem:cond3} and note that eqn.~\eqref{eq:normbound}  holds. So, we subtract $\Ib_\dimone$ from each side of eqn.~\eqref{eq:normbound} to get
	\begin{flalign}
	&~\left(\frac{2}{2+\zeta}-1\right)\Ib_\dimone\preccurlyeq\Qb^{-1/2}\Ab\Db^2\Ab^\ts\Qb^{-1/2}-\Ib_\dimone\preccurlyeq\left(\frac{2}{2-\zeta}-1\right)\Ib_\dimone\nonumber\\
	\Leftrightarrow&~-\frac{\zeta}{2+\zeta}\Ib_\dimone\preccurlyeq\Qb^{-1/2}\Ab\Db^2\Ab^\ts\Qb^{-1/2}-\Ib_\dimone\preccurlyeq\frac{\zeta}{2-\zeta}\Ib_\dimone\nonumber\\
	\Rightarrow&~-\frac{\zeta}{2-\zeta}\Ib_\dimone\preccurlyeq\Qb^{-1/2}\Ab\Db^2\Ab^\ts\Qb^{-1/2}-\Ib_\dimone\preccurlyeq\frac{\zeta}{2-\zeta}\Ib_\dimone\label{eq:ineq5}\\
	\Leftrightarrow&~\|\Qb^{-1/2}\Ab\Db^2\Ab^\ts\Qb^{-1/2}-\Ib_\dimone\|_2\le\frac{\zeta}{2-\zeta}\le\zeta\label{eq:ineq6}.
	\end{flalign}
Eqn.~\eqref{eq:ineq5} holds as $\frac{\zeta}{2+\zeta}\le\frac{\zeta}{2-\zeta}$ and the last inequality of eqn.~\eqref{eq:ineq6} follows from $\zeta<1$.
\end{proof}

\vspace{0.02in}\noindent\textbf{Satisfying eqn.~\eqref{eq:pdcond1}.} Note that the condition in eqn.~\eqref{eq:pdcond1} already holds from Lemma~\ref{lem:cond3}, as we use the exact same sketching matrix $\Wb\in\RR{n}{w}$.

\vspace{0.02in}\noindent\textbf{Satisfying eqn.~\eqref{eq:pdcond2}.} Using Lemma~\ref{lem:normbound} and applying Lemma~\ref{lem:recursive} recursively, we get
\begin{flalign}
\|\tilde{\fb}^{(t)}\|_2\le\zeta\|\tilde{\fb}^{(t-1)}\|_2\le\dots\le\zeta^t\|\tilde{\fb}^{(0)}\|_2=\zeta^t\|\Qb^{-\nicefrac{1}{2}}\pb\|_2\nonumber\,.
\end{flalign}

\section{Steepest Descent}\label{sxn:psd}

We now replace Step 4 of Algorithm~\ref{algo:PCG} (our proposed PCG solver) by preconditioned steepest descent. We again prove that our analysis of the proposed infeasible long-step IPM remains essentially the same.

First, we construct the sketching matrix $\Wb$ as discussed in Section~\ref{sxn:background}, with a slightly more stringent accuracy guarantee. More specifically, we necessitate that
\begin{flalign}
\nbr{\Vb^\ts \Wb \Wb^\ts \Vb - \Ib_{\dimone}}_2\le \frac{\zeta(1-\zeta)}{2}\label{eq:sd_struct}
\end{flalign}
holds with probability at least $1-\delta$ for a constant $\zeta \in [0,1]$. Notice that the sketching dimension $w=\Ocal(m\log(\nicefrac{m}{\delta}))$ and the running time needed to compute $\Qb^{\nicefrac{-1}{2}}$ (which is $\Ocal(\nnz{\Ab}\cdot \log(m/\delta)+m^3\log(m/\delta))$) remain, asymptotically, the same.
In the case of steepest descent, it turns out that at each iteration the search direction is the negative of the gradient, which is equal to the residual $\tilde{\fb}^{(j)}$.  Moreover, the step size $\alpha_j$ is determined by an exact \emph{line search} that minimizes the underlying quadratic function:
$$\alpha_j=\frac{\tilde{\fb}^{(j)^\ts}\tilde{\fb}^{(j)}}{\tilde{\fb}^{(j)^\ts}\Qb^{-\nicefrac{1}{2}}\,\Ab\Db^2\Ab^\ts\Qb^{-\nicefrac{1}{2}}\tilde{\fb}^{(j)}}.$$
For this choice of $\alpha_j$, it is easy to verify that the current gradient is orthogonal to the previous one. 

\begin{algorithm}[H]
	\caption{Steepest Descent Solver}\label{algo:sd_solver}
	\begin{algorithmic}
		\State \textbf{Input:}
		$\Ab\Db\in\RR{\dimone}{\dimtwo}$, $\pb\in\R{\dimone}$;
		number of iterations $t>0$;
		sketching matrix $\Wb \in \mathbb{R}^{\dimtwo\times w}$;
		\State \textbf{Initialize:}
		$\tilde{\zb}^{0} \gets \zero_\dimone $;
		
		\For {$j=0$ \textbf{to} $t-1$}
		
		\vspace{2mm}
		\State $\alpha_j=\frac{\tilde{\fb}^{(j)^\ts}\tilde{\fb}^{(j)}}{\tilde{\fb}^{(j)^\ts}\Qb^{-\nicefrac{1}{2}}\,\Ab\Db^2\Ab^\ts\Qb^{-\nicefrac{1}{2}}\tilde{\fb}^{(j)}}$;
		
		\vspace{1.5mm}
		\State $\tilde{\zb}^{j+1} \gets \tilde{\zb}^{j}-\alpha_j\tilde{\fb}^{(j)}$;
		\EndFor
		
		\vspace{1.5mm}
		\State \textbf{Output:} return $\tilde{\zb}^{t}$;
	\end{algorithmic}
\end{algorithm}

\noindent Similar to Lemma~\ref{lem:recursive}, our next result reveals a recursive relation between the search directions which will be instrumental in bounding $\tilde{\fb}^{(t)}$.
\begin{lemma}\label{lem:sd_recursive}
	Let $\tilde{\fb}^{(j)}$, $j=1\ldots t$ be the residual vectors at each iteration and let $\alpha_j$ be as in Algorithm~\ref{algo:sd_solver}. Then,
	\begin{flalign}
	\tilde{\fb}^{(j+1)}=~\left(\Ib_n-\alpha_j\Qb^{-\nicefrac{1}{2}}\Ab\Db^2\Ab^\ts\Qb^{-\nicefrac{1}{2}}\right)\tilde{\fb}^{(j)}\label{eq:err_sd}.
	\end{flalign}
	Recall that $\Qb=\Ab\Db\Wb\Wb^\ts\Db\Ab^\ts$ and $\tilde{\fb}^{(j)}=\Qb^{-\nicefrac{1}{2}}(\Ab\Db^2\Ab^\ts\Qb^{-\nicefrac{1}{2}}\tilde{\zb}^j-\pb)$\,.
\end{lemma}

\begin{proof}
	In Algorithm~\ref{algo:sd_solver}, we pre-multiply $\tilde{\zb}^{j+1}$  by $\Qb^{-\nicefrac{1}{2}}\Ab\Db^2\Ab^\ts\Qb^{-\nicefrac{1}{2}}$ and then subtract $\Qb^{-\nicefrac{1}{2}}\pb$ to get
	\begin{flalign}
	\tilde{\fb}^{(j+1)}=&~\Qb^{-\nicefrac{1}{2}}\Ab\Db^2\Ab^\ts\Qb^{-\nicefrac{1}{2}}\tilde{\zb}^{j+1}-\Qb^{-\nicefrac{1}{2}}\pb\nonumber\\=&~ \Qb^{-\nicefrac{1}{2}}\Ab\Db^2\Ab^\ts\Qb^{-\nicefrac{1}{2}}\tilde{\zb}^{j}-\Qb^{-\nicefrac{1}{2}}\pb-\alpha_j\Qb^{-\nicefrac{1}{2}}\Ab\Db^2\Ab^\ts\Qb^{-\nicefrac{1}{2}}\tilde{\fb}^{(j)}\nonumber\\
	=&~\tilde{\fb}^{(j)}-\alpha_j\Qb^{-\nicefrac{1}{2}}\Ab\Db^2\Ab^\ts\Qb^{-\nicefrac{1}{2}}\tilde{\fb}^{(j)}=~\left(\Ib_{\dimone}-\alpha_j\Qb^{-\nicefrac{1}{2}}\Ab\Db^2\Ab^\ts\Qb^{-\nicefrac{1}{2}}\right)\tilde{\fb}^{(j)}\nonumber\,,
	\end{flalign}
	which concludes the proof.
\end{proof}

\noindent Next, using this new condition in eqn.~\eqref{eq:sd_struct}, we  bound $\nbr{\Ib_{\dimone}-\alpha_j\Qb^{-\nicefrac{1}{2}}\Ab\Db^2\Ab^\ts\Qb^{-\nicefrac{1}{2}}}_2$.
\begin{lemma}\label{lem:alpha}
	If eqn.~\eqref{eq:sd_struct} is satisfied, then
	$\abs{\alpha_j-1}\le\frac{\zeta(1-\zeta)}{2}$.
\end{lemma}
\begin{proof}
	First, we rewrite eqn.~\eqref{eq:sd_struct} as follows,
	\begin{flalign*}
	-\frac{\zeta(1-\zeta)}{2}\Ib_{\dimone}\preccurlyeq\Vb^\ts\Wb\Wb^\ts\Vb-\Ib_{\dimone}\preccurlyeq\frac{\zeta(1-\zeta)}{2}\Ib_{\dimone}.
	\end{flalign*}
	Next, we pre- and post-multiply the above expression by $\Ub\Sigmab$ and $\Sigmab\Ub^\ts$ to get
	\begin{flalign}
	-\frac{\zeta(1-\zeta)}{2}\Ab\Db^2\Ab^\ts\preccurlyeq\underbrace{\Ab\Db\Wb\Wb^\ts\Db\Ab^\ts}_{\Qb}-\Ab\Db^2\Ab^\ts\preccurlyeq\frac{\zeta(1-\zeta)}{2}\Ab\Db^2\Ab^\ts.\label{eq:sd_struct2}
	\end{flalign}
	Now, pre and post-multiplying eqn.~\eqref{eq:sd_struct2} again by $\Qb^{-\nicefrac{1}{2}}$, we get
	\begin{flalign}
	&~\left(1-\frac{\zeta(1-\zeta)}{2}\right)\Qb^{-\nicefrac{1}{2}}\Ab\Db^2\Ab^\ts\Qb^{-\nicefrac{1}{2}}\preccurlyeq\Ib_{\dimone}\preccurlyeq\left(1+\frac{\zeta(1-\zeta)}{2}\right)\Qb^{-\nicefrac{1}{2}}\Ab\Db^2\Ab^\ts\Qb^{-\nicefrac{1}{2}}\nonumber\\
	\Rightarrow&~ \left(1-\frac{\zeta(1-\zeta)}{2}\right)\tilde{\fb}^{(j)^\ts}\Qb^{-\nicefrac{1}{2}}\Ab\Db^2\Ab^\ts\Qb^{-\nicefrac{1}{2}}\tilde{\fb}^{(j)}\le\tilde{\fb}^{(j)^\ts}\tilde{\fb}^{(j)}\le\left(1+\frac{\zeta(1-\zeta)}{2}\right)\tilde{\fb}^{(j)^\ts}\Qb^{-\nicefrac{1}{2}}\Ab\Db^2\Ab^\ts\Qb^{-\nicefrac{1}{2}}\tilde{\fb}^{(j)}\nonumber\\
	\Rightarrow&~
	\left(1-\frac{\zeta(1-\zeta)}{2}\right)\le\frac{\tilde{\fb}^{(j)^\ts}\tilde{\fb}^{(j)}}{\tilde{\fb}^{(j)^\ts}\Qb^{-\nicefrac{1}{2}}\Ab\Db^2\Ab^\ts\Qb^{-\nicefrac{1}{2}}\tilde{\fb}^{(j)}}\le\left(1+\frac{\zeta(1-\zeta)}{2}\right)\nonumber\\
	\Rightarrow&~\abs{\alpha_j-1}\le\frac{\zeta(1-\zeta)}{2}\,,\text{ for $j=1\ldots t$\,.}
	\end{flalign}
\end{proof}

\noindent Our next result shows that if eqn.~\eqref{eq:sd_struct} holds, then $\nbr{\Ib_{\dimone}-\alpha_j\Qb^{-\nicefrac{1}{2}}\Ab\Db^2\Ab^\ts\Qb^{-\nicefrac{1}{2}}}_2$ is upper bounded by a small quantity for $j=1\ldots t$.
\begin{lemma}\label{lem:sd_bound}
	If eqn.~\eqref{eq:sd_struct} is satisfied, then $\nbr{\Ib_{\dimone}-\alpha_j\Qb^{-\nicefrac{1}{2}}\Ab\Db^2\Ab^\ts\Qb^{-\nicefrac{1}{2}}}_2\le\zeta$\,, for $j=1\ldots t$.
\end{lemma}
\begin{proof}
	We note that eqn.~\eqref{eq:sd_struct} directly implies
	\begin{flalign}
	\nbr{\Vb^\ts \Wb \Wb^\ts \Vb - \Ib_{\dimone}}_2\le \frac{\zeta}{2}.\label{eq:sd_struct1}
	\end{flalign}
	Now, as eqn.~\eqref{eq:sd_struct1} holds, from  eqn.~\eqref{eq:normbound} in the proof of Lemma~\ref{lem:cond3},we have
	\begin{flalign}
	&~\left(1+\frac{\zeta}{2}\right)^{-1}\Ib_{\dimone}\preccurlyeq\Qb^{-1/2}\Ab\Db^2\Ab^\ts\Qb^{-1/2}\preccurlyeq\left(1-\frac{\zeta}{2}\right)^{-1}\Ib_{\dimone}\nonumber\\
	\Leftrightarrow&~\left(\frac{2\alpha_j}{2+\zeta}-1\right)\Ib_{\dimone}\preccurlyeq \alpha_j\Qb^{-1/2}\Ab\Db^2\Ab^\ts\Qb^{-1/2}-\Ib_{\dimone}\preccurlyeq\left(\frac{2\alpha_j}{2-\zeta}-1\right)\Ib_{\dimone}\nonumber\\
	\Leftrightarrow&~\frac{2(\alpha_j-1)-\zeta}{2+\zeta}\,\Ib_{\dimone}\preccurlyeq \alpha_j\Qb^{-1/2}\Ab\Db^2\Ab^\ts\Qb^{-1/2}-\Ib_{\dimone}\preccurlyeq\frac{2(\alpha_j-1)+\zeta}{2-\zeta}\,\Ib_{\dimone}\label{eq:sd_bound}.
	\end{flalign}
The above expression follows by multiplying eqn.~\eqref{eq:normbound} by $\alpha_j$ and then subtracting $\Ib_{\dimone}$.
Now, from Lemma~\ref{lem:alpha}, we get $-\zeta(1-\zeta)\le2(\alpha_j-1)\le\zeta(1-\zeta)$ for $j=1\ldots t$. Using this in eqn.~\eqref{eq:sd_bound}, we further have
	\begin{flalign}
	&~-\frac{\zeta(1-\zeta)+\zeta}{2+\zeta}\,\Ib_{\dimone}\preccurlyeq \alpha_j\Qb^{-1/2}\Ab\Db^2\Ab^\ts\Qb^{-1/2}-\Ib_{\dimone}\preccurlyeq\frac{\zeta(1-\zeta)+\zeta}{2-\zeta}\,\Ib_{\dimone}\nonumber\\
	\Leftrightarrow&~-\frac{\zeta(2-\zeta)}{2+\zeta}\,\Ib_{\dimone}\preccurlyeq \alpha_j\Qb^{-1/2}\Ab\Db^2\Ab^\ts\Qb^{-1/2}-\Ib_{\dimone}\preccurlyeq\zeta\,\Ib_{\dimone}\nonumber\\
	\Rightarrow&~-\zeta\,\Ib_{\dimone}\preccurlyeq \alpha_j\Qb^{-1/2}\Ab\Db^2\Ab^\ts\Qb^{-1/2}-\Ib_{\dimone}\preccurlyeq\zeta\,\Ib_{\dimone}\label{eq:normbound2}\\
	\Rightarrow&~\nbr{\Ib_{\dimone}-\alpha_j\Qb^{-\nicefrac{1}{2}}\Ab\Db^2\Ab^\ts\Qb^{-\nicefrac{1}{2}}}_2\le\zeta\nonumber\,,
	\end{flalign}
	where eqn.~\eqref{eq:normbound2} is due to the fact that $\frac{2-\zeta}{2+\zeta}\le 1$.
\end{proof}

\vspace{0.02in}\noindent\textbf{Satisfying eqn.~\eqref{eq:pdcond1}.} As eqn.~\eqref{eq:sd_struct1} holds, eqn.~\eqref{eq:pdcond1} directly follows from Lemma~\ref{lem:cond3}.

\vspace{0.02in}\noindent\textbf{Satisfying eqn.~\eqref{eq:pdcond2}.} Using Lemma~\ref{lem:sd_bound} and applying Lemma~\ref{lem:sd_recursive} recursively, we get
\begin{flalign}
\|\tilde{\fb}^{(t)}\|_2\le\zeta\|\tilde{\fb}^{(t-1)}\|_2\le\dots\le\zeta^t\|\tilde{\fb}^{(0)}\|_2=\zeta^t\|\Qb^{-\nicefrac{1}{2}}\pb\|_2\nonumber\,.
\end{flalign}

\section{Convergence analysis of Algorithm~\ref{algo:iipm}}\label{app:convergence}

\subsection{Additional notation}

For any two vectors $\ab=(a_1,\dots,a_\ell)^\ts$ and $\bb=(b_1,\dots,b_\ell)^\ts$ let $\ab\circ\bb=(a_1b_1,\dots,a_\ell b_\ell)^\ts$. For any vector $\ab\in\R{n}$ its $\ell_\infty$ norm is defined as $\|\ab\|_\infty=\max_i\abs{a_i}$. We heavily use the following standard inequality to prove results in this section:
\begin{flalign}\label{eq:normineq}
\abs{\frac{\ab^\ts\one_n}{n}}\le\|\ab\|_\infty\le\|\ab\|_2.
\end{flalign}

\subsection{Number of iterations for the CG solver}
In this section, most of the proofs follow \cite{Mon03} except for the fact that we used our sketching based preconditioner $\Qb^{\nicefrac{-1}{2}}$. Recall that $\mathcal{S}$ is the set of optimal and feasible solutions for the proposed LP.

\begin{lemma}\label{lem:ineq}
	Let $(\xb^{0},\yb^{0},\sbb^{0})$ be the initial point with $(\xb^{0},\sbb^{0})>\zero$ and $(\xb^{*},\yb^{*},\sbb^{*})\in\mathcal{S}$ such that $(\xb^{*},\sbb^{*})\le(\xb^{0},\sbb^{0})$ with $\sbb^{0}\ge|\Ab^\ts\yb^{0}-\cbb|$. Then, for any point $(\xb,\yb,\sbb)\in\mathcal{N}(\gamma)$ such that $\rb=\eta\,\rb^{0}$ and $0\le\eta\le\min\left\{1,\frac{\sbb^\ts\xb}{\sbb^{0\ts}\xb^{0}}\right\}$, we get
	\begin{subequations}
		\begin{flalign}
		&(i)~~\eta\,(\xb^\ts\sbb^{0}+\sbb^\ts\xb^{0})\le\,3\dimtwo\mu\label{eq:ineq}\,,\\
		&(ii)~~\eta\,\|\Sb(\xb^{*}-\xb^{0})\|_2\le\eta\,\|\Sb\xb^{0}\|_2\le\eta\sbb^{\ts}\xb^{0}\le\,3\dimtwo\mu\label{eq:ineq2}\,,\\
		&(iii)~~\eta\,\|\Xb(\sbb^{0}+\Ab^\ts\yb^{0}-\cbb)\|_2\le~2\eta\,\|\Xb\sbb^{0}\|_2\le~2\eta\,\xb^\ts\sbb^{0}\le~6\dimtwo\mu\label{eq:ineq3}\,.
		\end{flalign}
	\end{subequations}
\end{lemma}
\begin{proof}
	We prove eqns.~\eqref{eq:ineq}--\eqref{eq:ineq3} below.
	
\noindent\textbf{Proof of eqn.~\eqref{eq:ineq}.}
	For completeness, we provide a proof of eqn.~\eqref{eq:ineq} following~\cite{Mon03}. Since $(\xb^{*}, \sbb^{*}, \yb^{*}) \in \mathcal{S}$, the following equalities hold:
	\begin{subequations}
		\begin{flalign}
		\Ab\xb^{*} &=\bb\label{eq:cond1}\\
		\Ab^\ts\yb^{*}+\sbb^{*} &=\cbb.\label{eq:cond2}
		\end{flalign}
	\end{subequations}
	
	\noindent Furthermore, $\rb=\eta \rb^{0}$ implies
	\begin{subequations}
		\begin{flalign}
		\Ab\xb-\bb &=\eta(\Ab\xb^{0}-b)\label{eq:cond3}\\
		\Ab^\ts\yb+\sbb-\cbb &=\eta(\Ab^\ts\yb^{0}+\sbb^{0}-\cbb).\label{eq:cond4}
		\end{flalign}
	\end{subequations}
	
	\noindent Combining eqn.~\eqref{eq:cond1} with eqn.~\eqref{eq:cond3} and eqn.~\eqref{eq:cond2} with eqn.~\eqref{eq:cond4}, we get
	\begin{subequations}
		\begin{flalign}
		\Ab\big(\xb-\eta\xb^{0}-(1-\eta)\xb^{*}\big) &=\zero\label{eq:cond5}\\
		\Ab^\ts(\yb-\eta\yb^{0}-(1-\eta)\yb^{*})+(\sbb-\eta\sbb^{0}-(1-\eta)\sbb^{*}) &=\zero.\label{eq:cond6}
		\end{flalign}
	\end{subequations}
	Multiplying eqn.~\eqref{eq:cond6} by $\left(\xb-\eta \xb^{0}-(1-\eta)\xb^{*}\right)^\ts$ on the left and using eqn.~\eqref{eq:cond5},
	we get
	$$
	\left(\xb-\eta \xb^{0}-(1-\eta)\xb^{*}\right)^\ts\left(\sbb-\eta\sbb^{0}-(1-\eta)\sbb^{*}\right)=0.
	$$
	Expanding we get
	\begin{flalign}
	&\eta\left(\xb^{0^{\ts}}\sbb+\xb^{\ts}\sbb^{0}\right)=\eta^{2} \xb^{0^{\ts}}\sbb^{0}+(1-\eta)^{2} (\xb^{*})^{\ts} \sbb^{*}+\xb^{\ts}\sbb\nonumber\\
	&~~~~~~~~~~~~~~~~~~~~~~~~~~~~~+\eta(1-\eta)\left(\xb^{0^{\ts}}\sbb^{*}+(\xb^{*})^{\ts} \sbb^{0}\right)-(1-\eta)\left((\xb^{*})^{\ts}\sbb+\xb^{\ts} \sbb^{*}\right).\label{eq:cond7}
	\end{flalign}
	Next, we use the given conditions and rewrite eqn.~\eqref{eq:cond7} as
	\begin{flalign}
	\eta\left(\xb^{0^{\ts}} \sbb+\sbb^{0^{\ts}} \xb\right) & \leq \eta^{2} \xb^{0^{\ts}} \sbb^{0}+\xb^{\ts} \sbb+\eta(1-\eta)\left(\xb^{0^{\ts}} \sbb^{*}+\sbb^{0^{\ts}} \xb^{*}\right) \nonumber\\
	& \leq \eta^{2} \xb^{0^{\ts}} \sbb^{0}+\xb^{\ts} \sbb+2 \eta(1-\eta) \xb^{0^{\ts}} \sbb^{0} \nonumber\\
	& \leq 2 \eta \xb^{0^{\ts}} \sbb^{0}+\xb^{\ts} \sbb \leq 3 \xb^{\ts} \sbb=3\dimtwo\mu.\label{eq:fin_bound}
	\end{flalign}
	The first inequality in eqn.~\eqref{eq:fin_bound} follows from the following facts. First, $(1-\eta)((\xb^{*})^{\ts}\sbb+\xb^{\ts} \sbb^{*})\ge0$ as $(\xb^{*}, \sbb^{*}) \geq \zero$ and $(\xb^{0}, \sbb^{0}) \geq \zero$. Second, as $(\xb^{*}, \sbb^{*}, \yb^{*}) \in \mathcal{S}$ (which implies $\xb^{*}\circ\sbb^{*}=\zero$), we have $(\xb^{*})^{\ts} \sbb^{*}=0$. The second inequality in eqn.~\eqref{eq:fin_bound} holds as $\xb^{*} \leq \xb^{0}$, $\sbb^{*} \leq \sbb^{0}$,  $(\xb^{*}, \sbb^{*}) \geq \zero$, and $(\xb^{0}, \sbb^{0}) \geq \zero$; combining them we get $(\xb^{0^{\ts}} \sbb^{*}+\sbb^{0^{\ts}} \xb^{*})\le2\,\xb^{0^\ts}\sbb^{0}$. Third inequality in eqn.~\eqref{eq:fin_bound} is true as we have $ \eta^{2} \xb^{0^{\ts}}+2 \eta(1-\eta) \xb^{0^{\ts}} \sbb^{0}=2 \eta\xb^{0^{\ts}} \sbb^{0}-\eta^2\xb^{0^{\ts}} \sbb^{0}\le2 \eta\xb^{0^{\ts}} \sbb^{0}$. The final inequality holds as $\eta \leq \frac{\xb^{\ts} \sbb}{\xb^{0^\ts} \sbb^{0}}$.\qed
	
	\noindent\textbf{Proof of eqn.~\eqref{eq:ineq2}.}
	The last inequality follows from eqn.~\eqref{eq:ineq}. The second to last inequality is also easy to prove as
	\begin{flalign}
	\|\Sb\xb^{0}\|_2=\sqrt{\sum_{i=1}^{s}(s_ix_i^{0})^2}\le\sqrt{\left(\sum_{i=1}^{s}s_ix_i^{0}\right)^2}=\sbb^\ts\xb^{0}\,.
	\end{flalign}
	
	\noindent To prove the first inequality in eqn.~\eqref{eq:ineq2}, we use the fact $\xb^{0}\ge\xb^{*}$ as follows:
	\begin{flalign}
	\|\Sb\xb^{0}\|_2^2-\|\Sb(\xb^{*}-\xb^{0})\|_2^2=&~\sum_{i=1}^\dimtwo(s_ix_i^{0})^2-\sum_{i=1}^\dimtwo s_i^2\left((x_i^{*})^2+(x_i^{0})^2-2x_i^{*}x_i^{0}\right)\nonumber\\
	=&~\sum_{i=1}^\dimtwo s_i^2\left(2x_i^{*}x_i^{0}-(x_i^{*})^2\right)\ge 0\nonumber\,.
	\end{flalign}\qed
	
	\noindent \textbf{Proof of eqn.~\eqref{eq:ineq3}.}
	To prove this we use a similar approach as in eqn.~\eqref{eq:ineq2}. The last inequality directly follows from eqn.~\eqref{eq:ineq}; the second to last inequality is also easy to prove as
	\begin{flalign}
	\|\Xb\sbb^{0}\|_2=\sqrt{\sum_{i=1}^{\dimtwo}(x_is_i^{0})^2}\le\sqrt{\left(\sum_{i=1}^{\dimtwo}x_is_i^{0}\right)^2}=\xb^\ts\sbb^{0}\,.
	\end{flalign}
	For the first inequality, we proceed as follows:
	\begin{flalign}
	\|\Xb(\sbb^{0}+\Ab^\ts\yb^{0}-\cbb)\|_2^2=&~\|\Xb\sbb^{0}\|_2^2+\|\Xb(\Ab^\ts\yb^{0}-\cbb)\|_2^2+2\sbb^{0^\ts}\Xb^\ts\Xb(\Ab^\ts\yb^{0}-\cbb)\nonumber\\
	=&~\|\Xb\sbb^{0}\|_2^2+\sum_{i=1}^\dimtwo x_i^2(\Ab^\ts\yb^{0}-\cbb)_i^2+2\sum_{i=1}^\dimtwo x_i^2s_i^{0}(\Ab^\ts\yb^{0}-\cbb)_i\nonumber\\
	\le&~\|\Xb\sbb^{0}\|_2^2+\sum_{i=1}^\dimtwo (x_is_i^{0})^2+2\sum_{i=1}^\dimtwo(x_is_i^{0})^2\nonumber\\
	=&~\|\Xb\sbb^{0}\|_2^2+\|\Xb\sbb^{0}\|_2^2+2\|\Xb\sbb^{0}\|_2^2=4\|\Xb\sbb^{0}\|_2^2.\label{eq:note4}
	\end{flalign}
	The inequality in eqn.~\eqref{eq:note4} follows from $x_i\ge0$, $s_i^{0}\ge0$ and $\abs{(\Ab^\ts\yb^{0}-\cbb)_i}\le s_i^{0}$ for all $i=1\ldots \dimtwo$.
\end{proof}

\noindent Our next result bounds $\|\Qb^{-\nicefrac{1}{2}}\pb\|_2$ which is instrumental in proving the final bound.
\begin{lemma}\label{thm:boundf}
	Let $(\xb^{0},\yb^{0},\sbb^{0})$ be the initial point with $(\xb^{0},\sbb^{0})>\zero$ such that $\xb^{0}\ge\xb^{*}$ and $\sbb^{0}\ge\max\{\sbb^{*},|\cbb-\Ab^\ts\yb^{0}|\}$ for some $(\xb^{*},\yb^{*},\sbb^{*})\in\mathcal{S}$. Furthermore, let $(\xb,\yb,\sbb)\in\mathcal{N}(\gamma)$ with $\rb=\eta\,\rb^{0}$ for some $0\le\eta\le1$. If the sketching matrix $\Wb\in\RR{\dimtwo}{w}$ satisfies the condition in eqn.~\eqref{eq:pdcond1}, then
	\begin{flalign}
	\|\Qb^{-\nicefrac{1}{2}}\pb\|_2\le~\sqrt{2}\left(\frac{9\dimtwo}{\sqrt{1-\gamma}}+\sigma\sqrt{\frac{\dimtwo}{1-\gamma}}+\sqrt{\dimtwo}\right)\sqrt{\mu}\nonumber\,.
	\end{flalign}
	Recall that $\rb=(\rb_p,\rb_d)=(\Ab\xb-\bb,\Ab^\ts\yb+\sbb-\cbb)$ and $\rb^{0}=(\rb_p^{0},\rb_d^{0})=(\Ab\xb^{0}-\bb,\Ab^\ts\yb^{0}+\sbb^{0}-\cbb)$\,.
\end{lemma}

\begin{proof}
Note that after correcting the approximation error of the CG solver using $\vb$, the primal and dual residuals $\rb=(\rb_p,\rb_d)$ corresponding to an iterate $(\xb,\yb,\sbb)\in\mathcal{N}(\gamma)$ always lie on the line segment between zero and $\rb^{(0)}$. In other words, $\rb=\eta\rb^{(0)}$ always holds for some $\eta\in[0,1]$. This was formally proven in Lemma 3.3 of \cite{Mon03}.
	In order to bound $\|\Qb^{-\nicefrac{1}{2}}\pb\|_2$,  first we express $\pb$ as in eqn.~\eqref{eq:normal} and rewrite
	\begin{flalign}
	\Qb^{-\nicefrac{1}{2}}\pb
	=&~\Qb^{-\nicefrac{1}{2}}\left(-\rb_p-\sigma\mu\Ab\Sb^{-1}\one_\dimtwo+\Ab\xb-\Ab\Db^2\rb_d\right).\label{eq:recur3}
	\end{flalign}
	Then, applying the triangle inequality to $\|\Qb^{-\nicefrac{1}{2}}\pb\|_2$ in eqn.~\eqref{eq:recur3}, we get
	\begin{flalign}
	\|\Qb^{-\nicefrac{1}{2}}\pb\|_2\le \Delta_1+\Delta_2+\Delta_3+\Delta_4\label{eq:recur4}\,,
	\end{flalign}
	where
	\begin{flalign}
	\Delta_1=&~\|\Qb^{-\nicefrac{1}{2}}\rb_p\|_2\nonumber\,,\\
	\Delta_2=&~\sigma\mu\|\Qb^{-\nicefrac{1}{2}}\Ab\Db(\Xb\Sb)^{-\nicefrac{1}{2}}\one_\dimtwo\|_2\nonumber\,,\\
	\Delta_3=&~\|\Qb^{-\nicefrac{1}{2}}\Ab\Db\Db^{-1}\xb\|_2\,,\nonumber\\
	\Delta_4=&~\|\Qb^{-\nicefrac{1}{2}}\Ab\Db^2\rb_d\|_2\,.\nonumber
	\end{flalign}
	To bound $\Delta_1$, $\Delta_2$, $\Delta_3$ and $\Delta_4$ we heavily use the condition of eqn.~\eqref{eq:pdcond1}.
	
	\paragraph{Bounding $\Delta_1$.} Using $\rb_p=\eta\,\rb_p^{0}$, $\rb_p^{0}=\Ab\xb^{0}-\bb$ and $\bb=\Ab\xb^{*}$, we rewrite $\Delta_1$ as
	\begin{flalign}
	\Delta_1=&~\eta\,\|\Qb^{-\nicefrac{1}{2}}\Ab(\xb^{0}-\xb^{*})\|_2\nonumber\\
	=&~\eta\,\|\Qb^{-\nicefrac{1}{2}}\Ab\Db\Db^{-1}(\xb^{0}-\xb^{*})\|_2\nonumber\\
	\le&~\eta\,\|\Qb^{-\nicefrac{1}{2}}\Ab\Db\|_2\|\Db^{-1}(\xb^{0}-\xb^{*})\|_2\nonumber\\
	\le&~\sqrt{2}\eta\,\|\Db^{-1}(\xb^{0}-\xb^{*})\|_2\nonumber\\
	=&~\sqrt{2}\eta\,\|(\Xb\Sb)^{-\nicefrac{1}{2}}\Sb(\xb^{0}-\xb^{*})\|_2\nonumber\\
	\le&~\sqrt{2}\eta\,\|(\Xb\Sb)^{-\nicefrac{1}{2}}\|_2\,\|\Sb(\xb^{0}-\xb^{*})\|_2\,,\label{eq:del11}
	\end{flalign}
	where the above steps follow from submultiplicativity and  eqn.~\eqref{eq:pdcond1}. From eqn.~\eqref{eq:pdcond1}, note that we have $\|\Qb^{-\nicefrac{1}{2}}\Ab\Db\|_2\le\sqrt{2}$ as $\zeta\le 1$\,. Now, applying eqn.~\eqref{eq:ineq2} and $\|(\Xb\Sb)^{-\nicefrac{1}{2}}\|_2=\max_{1\le i \le \dimtwo}\frac{1}{\sqrt{x_is_i}}$, we further have
	\begin{flalign}
	\Delta_1\le&~\sqrt{2}\,\max_{1\le i \le \dimtwo}\frac{1}{\sqrt{x_is_i}}\cdot 3\dimtwo\mu\nonumber\\
	\le&~ 3\sqrt{2}\,\dimtwo\sqrt{\frac{\mu}{1-\gamma}}\label{eq:del12}\,,
	\end{flalign}
	where the last inequality follows from $(\xb,\yb,\sbb)\in\mathcal{N}(\gamma)$.
	
	\paragraph{Bounding $\Delta_2$.} Applying submultiplicativity, we get
	\begin{flalign}
	\Delta_2=&~\sigma\mu\,\|\Qb^{-\nicefrac{1}{2}}\,\Ab\Db\,(\Xb\Sb)^{-\nicefrac{1}{2}}\one_\dimtwo\|_2\nonumber\\
	\le&~\sigma\mu\,\|\Qb^{-\nicefrac{1}{2}}\,\Ab\Db\|_2\|(\Xb\Sb)^{-\nicefrac{1}{2}}\one_\dimtwo\|_2\nonumber\\
	\le&~\sqrt{2}\,\sigma\mu\,\|(\Xb\Sb)^{-\nicefrac{1}{2}}\one_\dimtwo\|_2\nonumber\\
	=&~\sqrt{2}\,\sigma\mu\,\sqrt{\sum_{i=1}^{\dimtwo}\frac{1}{x_i s_i}}
	\le~\sqrt{2}\,\sigma\mu\,\sqrt{\sum_{i=1}^{\dimtwo}\frac{1}{(1-\gamma)\mu}}\nonumber\\
	=&~\sqrt{2}\,\sigma\,\sqrt{\frac{\dimtwo\,\mu}{(1-\gamma)}}\label{eq:del2}\,,
	\end{flalign}
	where the second to last inequality follows from eqn.~\eqref{eq:pdcond1} and the last inequality holds as $(\xb,\yb,\sbb)\in\mathcal{N}(\gamma)$.
	
	\paragraph{Bounding $\Delta_3$.} Using $\Db=\Sb^{-\nicefrac{1}{2}}\Xb^{\nicefrac{1}{2}}$ and $\xb=\Xb\,\one_\dimtwo$ we get
	\begin{flalign}
	\Delta_3=&~\|\Qb^{-\nicefrac{1}{2}}\,\Ab\Db\,(\Sb^{\nicefrac{1}{2}}\Xb^{-\nicefrac{1}{2}})\,\Xb\,\one_\dimtwo\|_2\nonumber\\
	=&~\|\Qb^{-\nicefrac{1}{2}}\,\Ab\Db\,(\Sb\Xb)^{\nicefrac{1}{2}}\,\one_\dimtwo\|_2\nonumber\\
	\le&~\|\Qb^{-\nicefrac{1}{2}}\,\Ab\Db\|_2\|(\Sb\Xb)^{\nicefrac{1}{2}}\,\one_\dimtwo\|_2\nonumber\\
	\le&~\sqrt{2}\,\sqrt{\sum_{i=1}^{\dimtwo}x_i s_i}=~\sqrt{2\dimtwo\,\mu}\label{eq:del3}\,,
	\end{flalign}
	where the inequalities follow from submultiplicativity and eqn.~\eqref{eq:pdcond1}.
	
	\paragraph{Bounding $\Delta_4$.} Using $\rb_d=\eta\,\rb_d^{0}$, we have
	\begin{flalign}
	\Delta_4=&~\eta\|\Qb^{-\nicefrac{1}{2}}\,\Ab\,\Db^2\rb_d^{0}\|_2\nonumber\\
	\le&~ \eta\|\Qb^{-\nicefrac{1}{2}}\,\Ab\Db\|_2\|(\Xb\Sb)^{-\nicefrac{1}{2}}\Xb\rb_d^{0}\|_2\nonumber\\
	\le&~ \sqrt{2}\eta\,\|(\Xb\Sb)^{-\nicefrac{1}{2}}\Xb(\Ab^\ts\yb^{0}+\sbb^{0}-\cbb)\|_2\nonumber\\
	\le&~\sqrt{2}\eta\,\|(\Xb\Sb)^{-\nicefrac{1}{2}}\|_2\,\|\Xb(\Ab^\ts\yb^{0}+\sbb^{0}-\cbb)\|_2\,,\nonumber
	\end{flalign}
	where the above inequalities follow from submultiplicativity and eqn.~\eqref{eq:pdcond1}. Now, applying eqn.~\eqref{eq:ineq3} and $\|(\Xb\Sb)^{-\nicefrac{1}{2}}\|_2\le\frac{1}{\sqrt{(1-\gamma)\mu}}$, we further have
	\begin{flalign}
	\Delta_4\le~6\sqrt{2}\dimtwo\sqrt{\frac{\mu}{1-\gamma}}.\label{eq:del4}
	\end{flalign}
	
	\paragraph{Final bound.} Combining eqns.~\eqref{eq:recur4}, \eqref{eq:del12}, ,\eqref{eq:del2}, \eqref{eq:del3} and \eqref{eq:del4}, we get
	\begin{flalign}
	\|\Qb^{-\nicefrac{1}{2}}\pb\|_2\le~\sqrt{2}\left(\frac{9\dimtwo}{\sqrt{1-\gamma}}+\sigma\sqrt{\frac{\dimtwo}{1-\gamma}}+\sqrt{\dimtwo}\right)\sqrt{\mu}\,.
	\end{flalign}
	This concludes the proof of Lemma~\ref{thm:boundf}.
\end{proof}

\begin{lemma}\label{lem:conouter}
	Let the sketching matrix $\Wb$ satisfy the conditions of eqns.~\eqref{eq:pdcond1} and \eqref{eq:pdcond2}. Then, after $t\ge\frac{\log(\nicefrac{4\sqrt{6\dimtwo}\,\psi}{\gamma\sigma})}{\log(\nicefrac{1}{\zeta})}$ iterations of the CG solver in Algorithm~\ref{algo:PCG},
	$$\|\tilde{\fb}^{(t)}\|_2\le\frac{\gamma\sigma}{4\sqrt{\dimtwo}}\sqrt{\mu}~~~\text{and}~~~\|\vb\|_2\le\frac{\gamma\sigma}{4}\mu.$$
	Here $\psi=\left(\frac{9\dimtwo}{\sqrt{1-\gamma}}+\sigma\sqrt{\frac{\dimtwo}{1-\gamma}}+\sqrt{\dimtwo}\right)$ and $\tilde{\fb}^{(t)}=\Qb^{-\nicefrac{1}{2}}\Ab\Db^2\Ab^\ts\Qb^{-\nicefrac{1}{2}}\tilde{\zb}^t-\Qb^{-\nicefrac{1}{2}}\pb$ is the residual of the CG solver.
\end{lemma}

\begin{proof}
	Combining Lemma~\ref{thm:boundf} and the condition in eqn.~\eqref{eq:pdcond2}, we get
	\begin{flalign}\label{eq:bd}
	\|\tilde{\fb}^{(t)}\|_2\le\zeta^t\psi\sqrt{2\mu}.
	\end{flalign}
	Now, $\|\tilde{\fb}^{(t)}\|_2\le\frac{\gamma\sigma}{4\sqrt{\dimtwo}}\sqrt{\mu}$ holds if
	$\sqrt{2}\psi\,\zeta^t\sqrt{\mu}\le\frac{\gamma\sigma}{4\sqrt{\dimtwo}}\sqrt{\mu}$, which holds if $\left(\frac{1}{\zeta}\right)^t\ge\frac{4\sqrt{2\,\dimtwo}\,\psi}{\gamma\sigma}$. The last inequality holds for our choice of $t$. Next, combining Lemma~\ref{lem:v} and eqn.~\eqref{eq:bd} we get
	\begin{flalign*}
	\|\vb\|_2\le\sqrt{3n\mu}\,\|\tilde{\fb}^{(t)}\|_2\le\sqrt{6n}\,\zeta^t\psi\mu
	\end{flalign*}
	Therefore, $\|\vb\|_2\le\frac{\gamma\sigma\mu}{4}$ holds if
	$\sqrt{6n}\psi\,\zeta^t\psi\mu\le\frac{\gamma\sigma\mu}{4}$, which holds for our choice of $t$.
	Now, fixing $\gamma$, $\sigma$, and $\zeta$, after $t=\Ocal (\log\dimtwo)$ iterations of Algorithm~\ref{algo:PCG} the conclusions of the lemma hold.
\end{proof}

\subsection{Determining step-size, bounding the number of iterations, and proof of Theorem~\ref{thm:1}}

Assume that the triplet $(\hat{\Delx},\hat{\Dely},\hat{\Dels})$ satisfies eqns.~\eqref{eq:delxhat}, \eqref{eq:addl} and \eqref{eq:delshat}. We rewrite this system in the following alternative form:
\begin{subequations}\label{eq:iip_3}
	\begin{flalign}
	\Ab\hat{\Delta\xb}=&-\rb_p, \label{eq:iip_3_1}\\
	\Ab^\ts\hat{\Delta\yb}+\hat{\Delta\sbb}=&-\rb_d,  \label{eq:iip_3_2}\\
	\Xb\hat{\Delta\sbb}+\Sb\hat{\Delta\xb}=&-\Xb\Sb\,\one_\dimtwo+\sigma\mu\,\one_\dimtwo - \vb.  \label{eq:iip_3_3}
	\end{flalign}
\end{subequations}
Indeed, we now show how to derive eqns.~\eqref{eq:delxhat}, \eqref{eq:addl} and \eqref{eq:delshat} from eqn.~\eqref{eq:iip_3}. Pre-multiplying both sides of eqn.~\eqref{eq:iip_3_3} by $\Ab\Sb^{-1}$ and noting that $\Db^2=\Xb\Sb^{-1}$, we get
\begin{flalign}
&~\Ab\Db^2\hat{\Delta\sbb}+\Ab\hat{\Delx}=-\Ab\Xb\one_n+\sigma\mu\Ab\Sb^{-1}\one_n-\Ab\Sb^{-1}\vb\nonumber\\
\Rightarrow&~\Ab\Db^2\hat{\Delta\sbb}=-\Ab\xb+\rb_p+\sigma\mu\Ab\Sb^{-1}\one_n-\Ab\Sb^{-1}\vb.\label{eq:iip_31}
\end{flalign}
Eqn.~\eqref{eq:iip_31} holds as $\Ab\Xb\one_n=\Ab\xb$ and, from eqn.~\eqref{eq:iip_3_1}, $\Ab\hat{\Delx}=-\rb_p$. Next, pre-multiplying eqn.~\eqref{eq:iip_3_2} by $\Ab\Db^2$, we get
\begin{flalign}
&~\Ab\Db^2\Ab^\ts\hat{\Dely}+\Ab\Db^2\hat{\Dels}=-\Ab\Db^2\rb_d\nonumber\\
\Rightarrow &~\Ab\Db^2\Ab^\ts\hat{\Dely}=-\rb_p-\sigma\mu\Ab\Sb^{-1}\one_n+\Ab\xb-\Ab\Db^2\rb_d+\Ab\Sb^{-1}\vb=\pb+\Ab\Sb^{-1}\vb.\label{eq:iip_311}
\end{flalign}
The first equality in eqn.~\eqref{eq:iip_311} follows from eqn.~\eqref{eq:iip_31} and the definition of $\pb$. This establishes
eqn.~\eqref{eq:addl}. Eqn.~\eqref{eq:delshat} directly follows from eqn.~\eqref{eq:iip_3_2}. Finally, we get eqn.~\eqref{eq:delxhat} by pre-multiplying eqn.~\eqref{eq:iip_3_3} by $\Sb^{-1}$.

Next, we define each new point traversed by the algorithm as $(\xnew,\ynew, \snew)$, where
\begin{flalign}
(\xnew,\ynew,\snew) &= (\xb, \yb, \sbb) + \alpha(\hat{\Delta\xb},\hat{\Delta \yb},\hat{\Delta\sbb}), \\
\munew &= \xnew^\ts \snew/ \dimtwo, \\
\rb(\alpha) &= \rb\left(\xnew,\snew,\ynew\right).
\end{flalign}
The goal in this section is to bound the number of iterations required by Algorithm~\ref{algo:iipm}.
Towards that end, we bound the magnitude of the step size $\alpha$. First, we provide an upper bound on $\alpha$, which allows us to show that each new point $(\xnew,\snew,\ynew)$ traversed by the algorithm stays within the neighborhood $\neigh$.
Second, we provide a lower bound on $\alpha$, which allows us to bound the number of iterations required. We use multiple lemmas from~\cite{Mon03}, which we reproduce here, without their proofs.

First, we provide an upper bound on $\alpha$, ensuring that each new point $(\xnew,\ynew,\snew)$ traversed by the algorithm stays within the neighborhood $\neigh$.
\begin{lemma}[Lemma 3.5 of \cite{Mon03}]\label{lemmamaxalpha1}
Assume $(\hat{\Delx},\hat{\Dely},\hat{\Dels})$  satisfies eqns.~(\ref{eq:iip_3}) for some $\sig > 0$, $(\xb, \yb,\sbb) \in \neigh$ (for $\gamma \in (0,1)$), and $\|\vb\|_2 \leq \frac{\gamma \sig \mu}{4}$. Then, $(\xnew,\ynew,\snew) \in \neigh$ for every scalar $\alpha$ such that
\begin{flalign} \label{alphalemma1}
0 \leq \alpha \leq \min \left\{ 1, \frac{\gamma \sig \mu}{4 \delxdelsnorm}\right\} .
\end{flalign}
\end{lemma}
\noindent We now provide a lower bound on the values of $\alphaused$ and the corresponding $\mu(\alphaused)$; see Algorithm~\ref{algo:iipm}.
\begin{lemma}[Lemma 3.6 of~\cite{Mon03}]\label{lemmaminalpha1}
In each iteration of Algorithm~\ref{algo:iipm}, if $\|\vb\|_2\le\frac{\gamma\sigma\mu}{4}$, then the step size $\alphaused$ satisfies
\begin{flalign} \label{alphalemma2}
\alphaused \geq \min \left\{ 1,  \frac{ \min   \{ \gamma \sigma, (1 - \frac{5}{4} \sigma) \} \mu } {4  \|\hat{\Delta x} \circ \hat{\Delta s}\|_\infty}  \right\}
	\end{flalign}
	and
	\begin{flalign} \label{mulemma2}
	\mu(\alphaused)= \Big[ 1 - \frac{\alphaused}{2} (1-\frac{5}{4}\sig) \Big] \mu.
\end{flalign}
\end{lemma}
\noindent At this point, we have provided a lower bound (eqn.~(\ref{alphalemma2})) for the allowed values of the step size $\alphaused$. Next, we show that this lower bound is bounded away from zero. From eqn.~(\ref{alphalemma2}) this is equivalent to showing that $\delxdelsnorm$ is bounded.
\begin{lemma}[Lemma 3.7 of~\cite{Mon03} (slightly modified)]\label{lemmaminalpha2}
Let $(\xb^{0},\yb^{0},\sbb^{0})$ be the initial point with $(\xb^{0},\sbb^{0})>0$ and $(\xb^0,\sbb^0) \geq (\xb^{*}, \sbb^{*}) $ for some $(\xb^{*},\yb^{*},\sbb^{*}) \in \mathcal{S}$. Let $(\xb,\yb,\sbb) \in \neigh$ be such that $\rb = \eta \rb^{0}$ for some $\eta \in [0,1]$ and $\|\vb\|_2\le\frac{\gamma\sigma\mu}{4}$. Then, the search direction  $(\hat{\Delta \xb},\hat{\Delta \yb},\hat{\Delta \sbb})$ produced by Algorithm \ref{algo:iipm} at each iteration satisfies
\begin{flalign} \label{normsmax}
\max \{\| \Db^{-1} \hat{\Delta \xb}\|_2, \| \Db \hat{\Delta \sbb}\|_2 \} \le~\left(1+\frac{\sigma^2}{1-\gamma}-2\sigma\right)^{\nicefrac{1}{2}}\sqrt{\dimtwo\mu}
+\frac{6\dimtwo}{\sqrt{(1-\gamma)}}\sqrt{\mu}+\frac{\gamma\sigma}{4\,\sqrt{1-\gamma}}\sqrt{\mu}.
\end{flalign}
\end{lemma}
\noindent We should note here that the above lemma is slightly different than Lemma 3.7 of~\cite{Mon03}. Indeed, Lemma 3.7 of~\cite{Mon03} actually proves the following bound:
\begin{flalign} \label{normsmax2}
\max \{\| \Db^{-1} \hat{\Delta \xb}\|_2, \| \Db \hat{\Delta \sbb}\|_2 \}
\le~\left(1+\frac{\sigma^2}{1-\gamma}-2\sigma\right)^{\nicefrac{1}{2}}\sqrt{\dimtwo\mu}+\frac{6\dimtwo}{\sqrt{(1-\gamma)}}\sqrt{\mu}+\frac{\gamma\sigma}{4\sqrt{n}}\sqrt{\mu}\,.
\end{flalign}
Notice that there is slight difference in the last term in the right-hand side, which does not asymptotically change the bound. The underlying reason for this difference is the fact that~\cite{Mon03} constructed the vector $\vb$ differently. In our case, we need to bound $\|(\Xb\Sb)^{-1/2}\vb\|_2$, which we do as follows:
\begin{flalign}
\| (\Xb \Sb)^{-1/2}\vb\|_2\le\| (\Xb \Sb)^{-1/2}\|_2\,\|\vb\|_2\le \frac{1}{\min_i \sqrt{x_is_i}}\,\frac{\gamma\sigma\mu}{4}\label{eq:db1}\,,
\end{flalign}
where in the above expression we use the fact that $\| (\Xb \Sb)^{-1/2}\|_2=\frac{1}{\min_i\sqrt{x_is_i}}$. Now as $(\xb,\yb,\sbb)\in\mathcal{N}(\gamma)$, we further have $x_is_i\ge (1-\gamma)\mu$ for all $i=1\ldots n$. Combining this with eqn.~\eqref{eq:db1}, we get
\begin{flalign}
\| (\Xb \Sb)^{-1/2}\vb\|_2\le\frac{\gamma\sigma\mu}{4\sqrt{(1-\gamma)\mu}}=\frac{\gamma\sigma}{4\,\sqrt{1-\gamma}}\sqrt{\mu}.\label{eq:bd2}
\end{flalign}
On the other hand,~\cite{Mon03} had a different construction of $\vb$ for which $\|(\Xb\Sb)^{-1/2}\vb\|_2=\|\tilde{\fb}^{(t)}\|_2$ holds. Therefore they had the following bound:
$$\|(\Xb\Sb)^{-1/2}\vb\|_2=\|\tilde{\fb}^{(t)}\|_2\le\frac{\gamma\sigma}{4\sqrt{n}}\sqrt{\mu}.$$
The next lemma bounds the number of iterations that Algorithm~\ref{algo:iipm} needs when started with an infeasible point that is sufficiently positive.
\begin{lemma}[Theorem 2.6  of \cite{Mon03}] \label{theoremOuter}
Assume that the constants $\gamma$ and $\sigma$ are such that $\max\{\gamma^{-1},(1-\gamma)^{-1},\sigma^{-1},(1-\frac{5}{4}\sigma)^{-1}\}=\Ocal(1)$. Let the initial point $(\xb^{0},\sbb^{0},\yb^{0})$ satisfy $(\xb^{0}, \sbb^{0}) \geq (\xb^{*}, \sbb^{*} )$ for some $(\xb^{*}, \sbb^{*},\yb^{*}) \in \mathcal{S}$ and $\|\vb\|_2\le\frac{\gamma\sigma\mu}{4}$. Algorithm \ref{algo:iipm} generates an iterate $(\xb^{k}, \sbb^{k}, \yb^{k})$ satisfying $\mu_k \leq \epsilon \mu_0$ and $\| \rb^{k}\|_2 \leq \epsilon \| \rb^{0}\|_2$ after
$\mathcal{O}(\dimtwo^2 \log{\nicefrac{1}{\epsilon}})$ iterations.
\end{lemma}
\noindent Finally, Theorem~\ref{thm:1} follows from Lemmas~\ref{lem:conouter} and~\ref{theoremOuter}.

\section{Additional notes on experiments}\label{app:experiments}

\subsection{Support Vector Machines (SVMs)}
\label{app:svm}

The classical $\ell_1$-SVM problem is as follows. We consider the task of fitting an SVM to data pairs $S = \{ (x_i, y_i)\}_{i=1}^m$, where $x_i \in \mathbb{R}^n$ and $y_i \in \{ + 1, - 1\}$. Here, $m$ is the number of training points, and $n$ is the feature dimension. The SVM problem with an $\ell_1$ regularizer has the following form:
\begin{align} \label{svm2}
\underset{{ w  }}{\operatorname{minimize}}  \quad &  \|w \|_1  \\
\text{subject to} \quad& y_i (w^T x_i + b') \geq 1, \quad i=1\ldots m.  \nonumber 
\end{align}
This problem can be written as an LP by introducing the variables $w^+$ and $w^-$, where $w = w^+ - w^-$. The objective becomes  $\sum_{j=1}^n w^+_j + w^-_j$, and we constrain $w^+_i \geq 0$ and $w^-_i \geq 0$. Note that the size of the constraint matrix in the LP becomes $m \times (2n +1)$.

\subsection{Random data}
\label{app:rand}
We generate random synthetic instances of linear programs as follows. To generate $\Ab \in \mathbb{R}^{m \times n}$, we set ${a}_{ij} \sim_{i.i.d.}U(0,1)$ with probability $p$ and ${a}_{ij} = 0$ otherwise. We then add min$\{m,n\}$  i.i.d. draws from $U(0,1)$ to the main diagonal, to ensure each row of $\Ab$ has at least one nonzero entry. We set $\bb = \Ab \xb + 0.1\zb$, where $\xb$ and $\zb$ are random vectors drawn from $N(0,1)$. Finally, we set $c \sim N(0,1)$.

\subsection{Real-world data}
\label{app:real}

We used a gene expression cancer RNA-Sequencing dataset, taken from the UCI Machine Learning repository. It is part of the RNA-Seq (HiSeq) PANCAN data set~\cite{Weinstein2013} and is a random extraction of gene expressions from patients who have different types of tumors: BRCA, KIRC, COAD, LUAD, and PRAD. We considered the binary classification task of identifying BRCA versus other types.

We also used the DrivFace dataset taken from the UCI Machine Learning repository. In the DrivFace dataset, each sample corresponds to an image of a human subject, taken while driving in real scenarios. Each image is labeled as corresponding to one of three possible gaze directions: left, straight, or right. We considered the binary classification task of identifying two different gaze directions: straight, or to either side (left or right).

\end{appendices}

\end{document}